\documentclass[10pt,reqno]{article}
\usepackage[a4paper,inner=35mm,outer=35mm,top=35mm,bottom=35mm]{geometry}
\usepackage{times}
\usepackage{authblk}
\usepackage{amsfonts,amssymb,amsmath,amsthm}
\usepackage{enumerate}
\usepackage{graphicx,xcolor}
\usepackage{natbib}

\newtheorem{theorem}{Theorem}[section]
\newtheorem{assumption}[theorem]{Assumption}
\newtheorem{lemma}[theorem]{Lemma}

\newtheorem{corollary}[theorem]{Corollary}
\newtheorem{proposition}[theorem]{Proposition}
\newtheorem{definition}[theorem]{Definition}

\newtheorem{remark}[theorem]{Remark}

\numberwithin{equation}{section}
\newcommand{\cA}{\mathcal{A}}
\newcommand{\cB}{\mathcal{B}}
\newcommand{\cC}{\mathcal{C}}

\newcommand{\cF}{\mathcal{F}}
\newcommand{\cG}{\mathcal{G}}
\newcommand{\cH}{\mathcal{H}}
\newcommand{\cI}{\mathcal{I}}

\newcommand{\cL}{\mathcal{L}}

\newcommand{\cP}{\mathcal{P}}
\newcommand{\cQ}{\mathcal{Q}}
\newcommand{\cR}{\mathcal{R}}
\newcommand{\cS}{\mathcal{S}}

\def \E{\mathbb{E}}
\def \F{\mathbb{F}}

\def \K{\mathbb{K}}

\def \N{\mathbb{N}}
\def \P{\mathbb{P}}
\def \Q{\mathbb{Q}}
\def \R{\mathbb{R}}

\def \G{\mathbb{G}}

\DeclareMathOperator{\gr}{graph}
\DeclareMathOperator{\dom}{dom}
\DeclareMathOperator{\supp}{supp}

\DeclareMathOperator{\inter}{int}

\DeclareMathOperator{\conv}{conv}

\DeclareMathOperator{\lin}{span}
\DeclareMathOperator{\proj}{proj}

\DeclareMathOperator{\NAP}{NA^s(\cP)}
\DeclareMathOperator{\NAPo}{NA_{\Phi}^s(\cP)}
\DeclareMathOperator{\NAs}{NA_2(\cP)}
\DeclareMathOperator{\NA}{NA(\hat{\cP})}
\DeclareMathOperator{\NAt}{NA(\hat{\cP}_t(\omega))}

\newcommand{\unit}{e_d}
\newcommand{\adm}{\cH^K}
\newcommand{\admr}{\hat{\cH}^r}
\newcommand{\admc}{\hat{\cH}}
\newcommand{\sh}{\pi}
\newcommand{\shr}{\hat{\pi}^r}
\newcommand{\shc}{\hat{\pi}}

\makeatletter
\makeatletter
\renewcommand*{\@fnsymbol}[1]{\ensuremath{\ifcase#1\or *\or **\or 
\mathsection\or \mathparagraph\or \|\or **\or \dagger\dagger
\or \ddagger\ddagger \else\@ctrerr\fi}}
\makeatother

\begin{document}

\title{On the quasi-sure superhedging duality with frictions}
\author[1]{Erhan Bayraktar\thanks{Erhan Bayraktar is supported in part by the National Science Foundation under the grant DMS-1613170, and in part by the Susan M. Smith Chair.}}
\author[2]{Matteo Burzoni\thanks{Matteo Burzoni gratefully acknowledges support by the ETH foundation and the Hooke Research Fellowship from the University of Oxford.}}
\affil[1]{University of Michigan}
\affil[2]{University of Oxford}
\date{\today}

\maketitle
\abstract{We prove the superhedging duality for a discrete-time financial market with proportional transaction costs under model uncertainty.
Frictions are modeled through solvency cones as in the original model of \cite{Kab99} adapted to the quasi-sure setup of \cite{BN13}.
Our approach allows to remove the restrictive assumption of No Arbitrage of the Second Kind considered in \cite{BDT17}  and to show the duality under the more natural condition of No Strict Arbitrage.
In addition, we extend the results to models with portfolio constraints.}\\



\section{Introduction}
It is more often the rule, rather than the exception, that socio-economics phenomena are influenced by a strong component of randomness.
Starting from the pioneering work of Knight (see e.g. \cite{Knight:21}) a distinction between \emph{risk} and \emph{uncertainty} has been widely accepted with respect to the nature of such a randomness.
We often call a situation \emph{risky} if a probabilistic description is available (e.g. the toss of a fair coin).
In contrast, we call a situation \emph{uncertain} if it cannot be fully described in probabilistic terms.
Simple reasons could be the absence of an objective model (e.g. the result of a horse race; see \cite{BM2017} and the references therein) or the lack of information (e.g. the draw from an urn whose composition is unknown).
The classical literature in mathematical finance has been mainly focusing on risk and the attention to problems of Knightian uncertainty has been drawn only relative recently starting from \cite{Av95}. 
In particular, fundamental topics such as the theory of arbitrage and the related superhedging duality have been systematically studied in frictionless discrete-time markets in \cite{BZ17,BN13} in a quasi-sure framework and in \cite{AB13,BFHMO,CK16} in a pointwise framework.

Under risk, the classical model of a discrete-time market with proportional transaction costs has been introduced in \cite{Kab99}.
The model is described by a collection of cones $\K:=\{K_t\}_{t=0,\ldots,T}$ which determines: i) admissible strategies; ii) solvency requirements; iii) pricing mechanisms.
More precisely, the latter are called \emph{consistent price systems} and they are essentially martingale processes taking values in the dual cones $K^*_t$.
Instances of such models have been considered, in the uncertainty case, in \cite{BCKT17,BayraktarZhangMOR,BN16,Bu16,DS14a}, nevertheless, the problem of establishing a quasi-sure superhedging duality has remained open.
Recently, a first duality result was obtained in \cite{BDT17} using a randomization approach (see also \cite{ADOT16,BZ18,BTY18} for other applications).
The idea is to construct a fictitious frictionless price process $\hat{S}$ for which: i) the superhedging price of an option in the market with frictions coincides with the corresponding superhedging price in the frictionless one; ii) the class of martingale measures for $\hat{S}$  produces the same prices for the option as the class of consistent price systems for the original market.
When these two properties are achieved, the duality follows from the frictionless results of \cite{BN13}.
In order to perform this program, a crucial role is played by the assumption of No Arbitrage of the Second Kind ($\NAs$), which ensures that the construction of the fictitious frictonless market is automatically arbitrage-free.
$\NAs$ prescribes that if a position is quasi-surely solvent at time $t+1$, it must be quasi-sure solvent at time $t$.
Such a condition is quite restrictive as it fails in very basic examples of one-period markets, even though, no sure profit can be made by market participants (see \cite[Remark 11]{BayraktarZhangMOR}).

In this paper we do not require the strong assumption $\NAs$ and we show the superhedging duality under the more natural condition of No Strict Arbitrage ($\NAP$).
The latter ensures that it is not possible to make profits without taking any risk, thus, it generalizes the classical no-arbitrage condition in frictionless markets. 
From a technical perspective, we also do not assume other unnecessary hypothesis taken in \cite{BDT17}: i) we do not require that transaction costs are uniformly bounded, stated differently, the bid-ask spreads relative to a chosen num\'eraire are not necessarily subsets of $[1/c,c]$ for some $c>0$; ii) we do not require the technical assumption $K^*_t\cap \partial\R^d_+=\{0\}$ for any $t=0,\ldots T$.
From a modeling perspective, our approach allows to extend the previous results to models where a process of portfolio constraints $C:=(C_t)_{t=0,\ldots,T}$ defines the admissible strategies in the market.
To the best of our knowledge, these results are new even in the classical case where a reference probability measure $\P$ is fixed.
As in \cite{BDT17} we assume the so-called \emph{efficient friction hypothesis} and adopt a randomization approach.

We first construct a backward procedure similar to the one of \cite{BayraktarZhangMOR} and based on a dynamic programming approach (see also \cite{BS18} for an extensive study of the related \emph{martingale selection problem}).
This procedure yields a new collection of cones $\tilde{\K}^*=(\tilde{K}^*_t)_{t=0,\ldots T}$ which is, in general, different from the original $\K$ and it is shown to be non-empty under the condition  $\NAP$.
Notably, it is not possible to apply directly the results of \cite{BDT17} (or a straightforward adaptation of them) to $\tilde{\K}^*$.
Indeed, in general $\tilde{K}^*_t$ will only have an analytic graph as opposed to the Borel-measurability of $K^*_t$.
The Borel-measurability assumption is crucial in order to apply the results of \cite{BN13} in frictionless markets.
To overcome this difficulty we propose a new randomization method. We do not design the frictionless process $\hat{S}$ to take values in $\tilde{\K}^*$ but we instead consider a suitable class of probabilities $\hat{\cP}$ in order to have $\hat{S}_t\in\tilde{K}^*_t$ $\hat{\cP}$-q.s.at each time.
Similarly to \cite{BDT17}, we finally prove that the desired duality can be deduced from duality results in frictionless market.
In particular, we use here those of \cite{BZ17}, which takes into account possible portfolio constraints.

We conclude the introduction by specifying the frequently used notation and the setup. The superhedging duality is stated in Section \ref{sec:main}.
The construction of the fictitious frictionless market is the content of Section \ref{sec:Random}. Finally, we prove the main result in Section \ref{sec:mainproof} where we also show how it extends to semi-static trading. 
\paragraph{Notation.}
For a topological space $X$, $\cB_X$ is the Borel sigma-algebra.
$\mathfrak{P}(X)$ is the class of all probability measures on $(X,\cB_X)$ and $\delta_x$ denotes the Dirac measure on $x\in X$.
For a probability measure $\P$ and a set $\cR\subset \mathfrak{P}(X)$ we say that $\P\ll\cR$ if there exists $\tilde{\P}\in\cR$ such that $\P\ll\tilde{\P}$.
A property is said to hold $\cR$-q.s. if it holds for any $\P\in\cR$.
A map $U$, defined on $X$ and taking values in the power set of $\R^d$, is called a\emph{ multifunction }(or \emph{random set}) and it is denoted by $U:X\rightrightarrows \R^d$. 
For a sigma-algebra $\cG$ on $X$, a multifunction $U$ is called $\cG$-measurable if, for any open set $O\subset\R^d$, we have 
\[
\{x\in X: U(x)\cap O\neq \emptyset\}\in\cG.
\]
A map $f:X\rightarrow\R^d$ with $f(x)\in U(x)$, for any $x\in\dom U:=\{x\in X: U(x)\neq \emptyset\}$ is called a selector of $U$.
We denote by $\cL^0(\cG;U)$ the class of $\cG$-measurable selectors of $U$.
For $U:X_1\times X_2\rightrightarrows \R^d$ and $x\in X_1$ fixed, the notation $U(x;\cdot)$ refers to the random set $U$ viewed as a (multi)function of $X_2$.
Given a class of probabilities $\cR\subset \mathfrak{P}(X_2)$, the (conditional) quasi-sure support of $U(x;\cdot)$, denoted by $\supp_{\cR}U(x;\cdot)$, is the smallest closed set $F\subset\R^d$ such that $U(x;\cdot)\subset F$, $\cR$-q.s. For a collection of random sets $U:=(U_t)_{t=0}^T$ adapted to a given filtration $\G$, we denote by $\cL^0(\G_-;U)$ the class of processes $H$ such that $H_{t+1}\in \cL^0(\cG_t;U_t)$ for every $t=0,\ldots T-1$.
Finally, for two $\R^d$-valued processes $H$ and $S$, we define $(H\circ S)_t:=\sum_{u=0}^{t-1} H_{u+1}\cdot(S_{u+1}-S_u)$.

\paragraph{Setup.}
Let $T\in\N$ be a fixed time horizon and $\cI:=\{0,\ldots T\}$. For later use we also define $\cI_{-1}:=\{-1,\ldots,T-1\}$. 
We consider a filtered space $(\Omega,\cF,\cF^u,\F,\F^u)$ endowed with a (possibly non-dominated) class of priors $\cP\subset\mathfrak{P}(\Omega)$ described as follows.
\begin{itemize}
	\item[-] $\Omega_1$ is a given Polish space. We choose $\Omega:=\Omega_{T}$, where $\Omega_t$ denotes the $(t+1)$-fold product of $\Omega_1$. Any $\omega\in\Omega_t$ is denoted by $\omega=(\omega_0,\ldots,\omega_t)$ with $\omega_s\in\Omega_1$ for any $0\le s\le t$.
	
	\item[-] We set $\cF:=\cB_\Omega$ and $\cF^u$ its universal completion. Similarly, the filtrations $\F=\{\cF_t\}_{t\in\cI}$ and $\F^u=\{\cF^u_t\}_{t\in\cI}$ are given by $\cF_t:=\cB_{\Omega_t}$ and $\cF^u_t$ its universal completion.
	\item[-] for each $t\in\cI$, $\cP_t$ is a random set of probabilities on $\Omega_1$ with analytic graph and $\cP_0$ is non-random. We set,
	\[\cP=\{P_0\otimes\cdots\otimes P_{T-1}: P_t\in\cL^0(\cF^u_t;\cP_t),  \forall t\in\{0,\ldots,T-1\}\},\]
	The class $\cL^0(\cF^u_t;\cP_t)$ is non-empty from Jankov-Von Neumann Theorem (see {\cite[Proposition 7.49]{BS78}}) so that $\cP$ is well defined through Fubini's Theorem.
\end{itemize}

\section{Main result}\label{sec:main}
We consider the general model of financial markets with proportional transaction costs introduced in \cite{Kab99}.
The model is fully described by a collection of random convex closed cones $\K:=\{K_t\}_{t\in\cI}\subset\R^d$ with $d\ge 2$, called \emph{solvency cones}.
These represent the sets of positions, in terms of physical units of $d$ underlying assets, which can be liquidated in the zero portfolio at zero cost. 
We assume that any position with non-negative coordinates is solvent, i.e., $\R^d_+\subset K_t$.
The set $-K_t$ represents the class of portfolios which are available at zero cost.
We assume that $K_t$ is $\cF_t$-measurable for any $t\in\cI$ with $K_0$ non-random.
Following standard notation, for a cone $K\subset\R^d$ we denote by $K^*:=\{x\in\R^d: x\cdot k\ge 0\, \forall k\in K\}$ its dual cone and by $K^\circ:=-K^*$ its polar cone. 

We generalize the model of \cite{Kab99} by introducing some constraints on the admissible positions in the market. 
These are represented by a collection  of random convex closed cones $C:=\{C_t\}_{t\in\cI}\subset\R^d$ such that every $C_t$ is $\cF_t$-measurable.
A zero-cost strategy $\eta:=(\eta_t)_{t\in\cI}$ is said to be \emph{admissible} if it satisfies  $\eta_t\in A_t$ for any $t\in\cI$, where 
\[
A_t:=\left\{\xi\in \cL^0(\cF^u_t;C_t): \xi=\sum_{s=0}^t-k_s\quad\text{with}\quad k_s\in K_s,\ \cP\text{-q.s.}\ \ \forall\ 0\le s\le t\right\}.
\]
In words, $\eta$ satisfies the constraints impose by $(C_t)_{t\in\cI}$ and it is obtained as the sum of portfolios which are available at zero cost. 
We denote by $\adm$ the class of admissible strategies and omit the dependence on $C$ as it will be fixed throughout the paper.

\begin{assumption}\label{ass:EF}
	We assume that $\inter(K^*_t)\neq\emptyset$ for any $t\in\cI$. Moreover, we assume that $C_t\subset C_{t+1}$ for any $t=0,\ldots, T-1$.
\end{assumption}
The first assumption is known as \emph{efficient friction hypothesis}.
The second one means that it is allowed to not trade between two periods.
The latter is obviously satisfied in the unconstrained case, that is $C_t\equiv\R^d$ for any $t\in\cI$.
\begin{definition}[$\NAP$] The \emph{No Strict Arbitrage condition} holds if, for all $t\in\cI$, $ A_t\cap\cL^0(\cF^u_t;K_t)=\{0\}$.
\end{definition}
This condition is the straightforward generalization to the quasi-sure setting of the classical one (see, e.g, \cite{KRS03} and the recent paper \cite{KM19} for a slightly weaker variant of this concept).
\begin{definition}[SCPS]\label{def:CPS}
	A couple $(Z,\Q)$ with $\Q\ll\cP$ is called a \emph{(strictly) consistent price system} if $Z_t\in \inter(K^*_t)$, $\Q$-a.s. $\forall t\in\cI$ and $H\circ Z$ is a $\Q$-local-supermartingale $\forall H\in\cL^0(\F^u_-;C)$.
\end{definition}
The interpretation is that $(Z,\Q)$ defines a frictionless arbitrage-free price process which is compatible with the model of transaction costs defined by $\{(K_t,C_t)\}_{t\in\cI}$. We shortly denote by $\cS$ the set of SCPS and by $\cS^0$ the class of normalized SCPS, namely, those satisfying $Z^d_t=1$ for any $t\in\cI$.

We are now ready to state the main result of the paper.
Let $G:\Omega\rightarrow\R^d$ be a Borel measurable random vector which represents the terminal payoff of an option in terms of physical units of the underlying assets.
The superhedging price of $G$ is given by
\begin{equation}\label{eq:defSH}
\sh_\K(G):=\inf\left\{y\in\R:\exists \eta\in\adm\text{ such that }\ y\unit+\eta_{T}-G\in K_T,\ \cP\text{-q.s.}\right\},
\end{equation}
where $\unit$ is the $d^{th}$ vector of the canonical basis of $\R^d$. 
\begin{theorem}\label{thm:main}
	Assume $\NAP$. For any Borel-measurable random vector $G$,
	\begin{equation}\label{eq:mainSH}
	\sh_\K(G)= \sup_{(Z,\Q)\in\cS}\E_\Q[G\cdot Z_T].
	\end{equation}
	Moreover, the superhedging price is attained when $\sh_\K(G)<\infty$.
\end{theorem}
The proof of Theorem \ref{thm:main} is given in Section \ref{sec:mainproof}.
The main difficulty is to establish Theorem \ref{thm:main} when only dynamic trading is allowed.
In Theorem \ref{thm:options} below we extend the duality to the case where also buy and hold positions in a finite number of options are allowed.
In the following it would be more convenient to extend the original market with an extra unconstrained component.
More precisely, one could consider the market $\bar{\K}$ with $\bar{K}_t:=K_t\times\R_+$ and $\bar{C}_t:=C_t\times\R$ for $t\in\cI$, which also satisfies Assumption \ref{ass:EF}.
It is easy to see that $\sh_\K(G)=\sh_{\bar{\K}}(\bar{G})$ with $\bar{G}=[G;0]$. 
On the dual side, $i:\cS\to\bar{\cS}^0$ with $i(Z,\Q)=([Z;1],\Q)$ is clearly a bijection and $ \E_{\Q}[G\cdot Z_T]=\E_{\Q}[\bar{G}\cdot [Z_T;1]]$.
\begin{remark}\label{rmk:unconstrained}
	Without loss of generality, we may assume that $(K_t,C_t)$ are already in the form described above, for any $t\in\cI$.
\end{remark}
We adapt some results of \cite{BayraktarZhangMOR} to the case of portfolio constraint.
These will be useful in the next sections.
Consider the collection of random sets $\tilde{\K}:=\{\tilde{K}_t\}_{t\in\cI}$, defined via a backward recursion as follows. We let $\tilde{K}^*_T:=K^*_T$ and
\begin{equation}\label{def:tildeK}
\tilde{K}^*_t:=K^*_t\cap \big( \overline{\conv(\Gamma_t)}+C^*_t\big),\quad t=T-1,\ldots,0,
\end{equation}
where, for $\omega\in\Omega_t$ fixed, $\Gamma_t(\omega):=\supp_{\cP_t(\omega)}\tilde{K}^*_{t+1}(\omega;\cdot)$. 
We define $\tilde{K}_t$ as the dual of $\tilde{K}^*_t$ for any $t\in\cI$.
The following are generalizations of Lemma 6 and Proposition 4 in \cite{BayraktarZhangMOR} to the present setting.
The proofs are analogous and we postpone them to the Appendix.
\begin{lemma}\label{lem:analytic_construction}
	$\tilde{K}^*_t$ has analytic graph for every $t\in\cI$.
\end{lemma}
\begin{proposition}\label{prop:nonempty} If $\K$ satisfies Assumption \ref{ass:EF} and $\NAP$, the same holds for $\tilde{\K}$. In particular, $\inter(\tilde{K}^*_t)\neq\emptyset$ $\cP$-q.s. for all $t\in\cI$.
\end{proposition}

\section{The randomization approach}\label{sec:Random}
In this section we construct an enlarged measurable space $(\hat{\Omega},\hat{\cF},\hat{\cF}^u,\hat{\F},\hat{\F}^u)$ endowed with a suitable class of probabilities $\hat{\cP}$.
On this space, we construct a price process $\hat{S}=(\hat{S}_t)_{t\in\cI}$ which represents a \emph{frictionless} financial market with the property that $\hat{S}_t\in \tilde{K}^*_t$ $\hat{\cP}$-q.s. for any $t\in\cI$ (Corollary \ref{cor:supports} below) and which is arbitrage free (Proposition \ref{prop: RNA implies NA} below).

We choose $\hat{\Omega}_1:=\Omega_1\times \R^{d-1}$ and set $\hat{\Omega}=\hat{\Omega}_{T}$, where $\hat{\Omega}_t$ denotes the $(t+1)$-fold product of $\hat{\Omega}_1$.
We endow $\hat{\Omega}$ with the filtration $\hat{\F}:=\{\hat{\cF}_t\}_{t\in\cI}$, where, for every $t\in\cI$, $\hat{\cF}_t:=\cF_t\otimes\cB_{\R^{d-1}}$. We denote by $\hat{\F}^u$ the universal completion of $\hat{\F}$.
Similarly, $\hat{\cF}:=\cB_{\hat{\Omega}}$ and $\hat{\cF}^u$ is its universal completion.
We shortly denote $(\omega,\theta)\in \hat{\Omega}_t$ for an element of the form $(\omega_0,\ldots,\omega_t,\theta_0,\ldots,\theta_t)$ with $\omega_s\in\Omega_1$ and $\theta_s\in\R^{d-1}$, for any $s=0,\ldots, t$. 
The collection of constraints extends to $\hat{\Omega}$ in the obvious way.
Since there is no source of confusion, we still denote them by $C=\{C_t\}_{t\in\cI}$.
We next construct the price process $\hat{S}$.
Recall that, for any $t\in\cI$, $K_t$ is Borel-measurable and, thus, also $K^*_t$.
Moreover, $\inter(K^*_t)$ is non-empty by Assumption \ref{ass:EF}.
From \cite[Lemma A.1]{BN16} there exists $S_t\in\cL^0(\cF_t;\inter(K^*_t))$.
Since $K^*_t\subset\R^d_+$ we can normalize $S_t$ with respect to, e.g., the last component, so that $S_t$ takes values in
\begin{equation}\label{def:K0}
K^{*,0}_t:=\{y\in K^*_t: y^d=1\},\qquad t\in\cI.
\end{equation}
We define a Borel-measurable price process $\hat{S}$ as,
\begin{equation}\label{eq:shadow}
\hat{S}_t(\omega,\theta)=[S_t^1(\omega)\theta^1_t,\ldots, S_t^{d-1}(\omega)\theta^{d-1}_t, 1],\qquad (\omega,\theta)\in\hat{\Omega},\ t\in\cI,
\end{equation}
where the last component serves as a num\'eraire.
The rest of the section is devoted to the construction of the desired set of probability measures $\hat{\cP}$.
For every $t\in\cI$, we define the random sets
\begin{equation}\label{def:theta_t}
\Theta_t(\omega):=\{\theta\in\R^{d-1}: \hat{S}_{t}(\omega,\theta)\in\inter(\tilde{K}^{*}_{t}(\omega))\},\quad \omega\in\Omega_t.
\end{equation}
\begin{lemma}\label{lem:construction}
	For every $t\in\cI$, $\Theta_t$ has an analytic graph.
\end{lemma} 
\begin{proof}
	In the proof we will repeatedly use the fact that the class of analytic sets is closed under countable union and intersection and that the image of an analytic set through a Borel-measurable function is again analytic.
	
	\emph{Step 1.} 
	For any $t\in\cI$, consider the random set $\tilde{K}^{*,d-1}_t:=\proj_{\R^{d-1}}(\tilde{K}^{*,0}_t)$ where the projection is taken over the first $d-1$ coordinates and $\tilde{K}^{*,0}_t$ is the analogous of \eqref{def:K0} for $\tilde{\K}$. 
	Observe that,
	\begin{equation}\label{eq:graphProj}
	\gr(\tilde{K}^{*,d-1}_t)=\proj_{\Omega_t\times\R^{d-1}}\left(\gr(\tilde{K}^{*}_t)\cap\big(\Omega_t\times\R^{d-1}\times\{1\}\big)\right).
	\end{equation}
	From  Lemma \ref{lem:analytic_construction}, $\gr(\tilde{K}^{*}_t)$ is analytic and so is the intersection in \eqref{eq:graphProj}.
	As the projection is a continuous map, we conclude that $\gr(\tilde{K}^{*,d-1}_t)$ is analytic.
	
	\emph{Step 2.} 
	We now show that the set
	\[
	A^{v}_t:=\left\{(\omega,\theta)\in\Omega_t\times\R^{d-1}: \theta^i_t=\frac{y^i}{S^i_t(\omega)},\quad y\in\tilde{K}^{*,d-1}_t(\omega)+v\right\}
	\]
	is analytic, for an arbitrary $v\in\R^{d-1}$ and $t\in\cI$. 
	This together with Lemma \ref{lem:interConvex} in the Appendix yields the claim, as $\gr(\Theta_t)$ is the intersection of countably many analytic sets of the from $A^v_t$.
	
	Observe that the function $f:\Omega_t\times\R^{d-1}\times\Omega_t\times\R^{d-1}\mapsto \Omega_t\times\R^{d-1}$ defined as
	\[
	f(\omega,y,\tilde{\omega},s)=
	\begin{cases}
	\left(\omega,\dfrac{y^1}{s^1},\ldots,\dfrac{y^{d-1}}{s^{d-1}}\right)&\omega=\tilde{\omega}\\
	\left(\omega,-1,\ldots,-1\right)&\omega\neq\tilde{\omega}
	\end{cases}
	\]
	is Borel-measurable.
	Recalling that $S_t>0$ and $K^{*}_t\subset\R^d_+$, we have that
	\[
	A_t^v=f\big(\gr(\tilde{K}^{*,d-1}_t+v),\gr(\tilde{S}_t)\big)\cap \big(\Omega_t\times\R^{d-1}_+\big),
	\]
	where $\tilde{S}$ is the process given by the first $d-1$ components of $S$. Since $\tilde{S}$ is Borel-measurable, $\gr(\tilde{S}_t)$ is a Borel set. Moreover, from Step 1, $\gr(\tilde{K}^{*,d-1}_t+v)$ is an analytic set. As $f$ is Borel-measurable, we conclude that $A_t^v$ is analytic.
	 \end{proof}
\begin{corollary}\label{cor:deltas}
	For any $t\in\cI$, the random set 
	\[
	\delta_{\Theta_t}(\omega):=\{\delta_\theta\in\mathfrak{P}(\R^{d-1}): \theta\in\Theta_t(\omega)\}
	\]
	has analytic graph.
\end{corollary}
\begin{proof}The graph of $\delta_{\Theta_t}$ is the image of the graph of $\Theta_t$ via the map $(\omega,\theta)\mapsto (\omega,\delta_\theta)$, which is an embedding
	(see {\cite[Theorem 15.8]{aliprantis}}). 
	Since the image of an analytic set through a continuous function is again analytic, the thesis follows.
	 \end{proof}
For $t\in\{0,\ldots T-1\}$ we define the random sets
\begin{equation}\label{eq:def_random sets}
\hat{\cP}_t(\omega):=\{\hat{\P}\in \mathfrak{P}(\hat{\Omega}_1) : \hat{\P}_{|\Omega_1}\in\cP_t(\omega),\ \hat{\P}(\gr\Theta_{t+1}(\omega;\cdot))=1 \},\quad \omega\in\Omega_t.
\end{equation}
We extend the definition to $t=-1$ with $\hat{\cP}_{-1}:=\{\hat{\P}\in \mathfrak{P}(\hat{\Omega}_1) : \hat{\P}(\gr\Theta_{0})=1 \}$, which is a non-random set as $K_0$ itself is non-random.

\begin{proposition}\label{prop:construction}
	The random sets $\hat{\cP}_t$ defined in \eqref{eq:def_random sets} have analytic graphs.
\end{proposition} 
\begin{proof}
	Fix $t\in\cI_{-1}$.
	For ease of notation, denote $A_{t+1}:=\gr\Theta_{t+1}$, which is analytic from Lemma \ref{lem:construction}.
	The function $\mathbf{1}_{A_{t+1}}:\Omega_t\times\hat{\Omega}_1\rightarrow\R$ is, thus, upper semianalytic.
	It is not difficult to show that the function $\phi:\Omega_t\times \mathfrak{P}(\hat{\Omega}_1)\rightarrow\R$ such that $\phi(\omega,\hat{\P})=\E_{\hat{\P}}[\mathbf{1}_{A_{t+1}}(\omega;\cdot)]$ is upper semianalytic (see proof of \cite[Lemma 4.10]{BN13}).
	As a consequence the set 
	\[
	\{(\omega,\hat{\P})\in\Omega_t\times \mathfrak{P}(\hat{\Omega}_1) : \hat{\P}(\gr\Theta_{t+1}(\omega;\cdot))=1\}=\phi^{-1}([1,\infty)),
	\]
	is analytic as $\phi$ is upper semianalytic.
	In particular, the thesis follows for $t=-1$.
	
	Let now $0\le t\le T-1$.
	Recall that the map $\pi_{\Omega_1}:\mathfrak{P}(\hat{\Omega}_1)\rightarrow \mathfrak{P}(\Omega_1)$, which associates to every $\hat{\P}\in \mathfrak{P}(\hat{\Omega}_1)$ its marginal on $\Omega_1$, is Borel measurable (see {\cite[Theorem 15.14]{aliprantis}}).
	Therefore, as in the proof of \cite[Lemma 2.12 (i)]{BDT17}, the set 
	\[
	\{(\omega,\hat{\P})\in\Omega_t\times\mathfrak{P}(\hat{\Omega}_1): \hat{\P}_{|\Omega_1}\in\cP_t(\omega)\}
	\]
	is also analytic.
	To conclude observe that $\gr(\hat{\cP}_t)$ is the intersection of the two previous sets.
	 \end{proof}
We now show that this class is non-empty on a sufficiently rich set of events.
\begin{lemma} \label{lem:non-empty}
	Assume $\NAP$.
	The set $N_t:=\{\omega\in\Omega_t: \hat{\cP}_t(\omega)=\emptyset\}$ is a universally measurable $\cP$-polar set, for any $t\in\cI_{-1}$. In particular, the same holds for the set $N:=\cup_{t\in\cI_{-1}}N_t$.
\end{lemma}
\begin{proof}
	Fix $t\in\cI_{-1}$. For $t=-1$ there is nothing to show as $\Theta_0$ is non-random and non-empty from Proposition \ref{prop:nonempty}. 
	Suppose $0\le t\le T-1$.
	Let $D_{t+1}:=\dom(\Theta_{t+1})$ which is analytic from Lemma \ref{lem:construction}. 
	As in the proof of Proposition \ref{prop:construction}, the function $\phi:\Omega_t\times \mathfrak{P}(\Omega_1)\rightarrow\R$ such that $\phi(\omega,\P)=\E_{\P}[\mathbf{1}_{D_{t+1}}(\omega;\cdot)]$ is upper semianalytic.
	We deduce that the set 
	\begin{equation}\label{eq:Bt}
	B_t:=\{(\omega,\P)\in\Omega_t\times \mathfrak{P}(\Omega_1): \phi(\omega,\P)=1\}\cap \gr(\cP_t)  
	\end{equation}
	is analytic and, thus, also its projection on $\Omega_t$.
	Denote by $N'_t:=(\proj_{\Omega_t}(B_t))^C\in\cF^u$.
	We show that, under $\NAP$, $N'_t$ is $\cP$-polar. 
	To see this observe that 
	\[
	D_{t+1}^C=\{\omega\in\Omega_{t+1}: \inter(\tilde{K}^{*}_{t+1}(\omega))= \emptyset\},
	\]
	is $\cP$-polar from Proposition \ref{prop:nonempty}.
	Suppose that there exists $\P\in\cP$ such that $\P(N'_t)>0$ and denote by $\{P_t\}_{t\in\cI}$ its disintegration.
	By definition of $B_t$, $P_t(\omega,D_{t+1}^C)>0$ for every $\omega\in N'_t$, therefore, the random variable
	\[
	\int_{\Omega_1} \mathbf{1}_{D_{t+1}^C}(\omega;\omega')P_t(\omega; d\omega'),\qquad \omega\in\Omega_t,
	\]
	is strictly positive on $N'_t$.
	Since $\P(N'_t)>0$, by integrating over $P_0\otimes P_{t-1}$ we obtain $\P(D_{t+1}^C)>0$. This is a contradiction, since $D_{t+1}^C$ is $\cP$-polar.
	
	It remains to show that $N'_t=N_t$.
	The inclusion $\subset$ follows from the definition of $B_t$.
	Take now $P_t$ an $\cF^u_t$-measurable selector of $B_t$ and $\delta_{\theta_{t+1}}\in\cL^0(\cF^u_{t+1};\delta_{\Theta_{t+1}})$, where $\delta_{\Theta_{t+1}}$ is defined in Corollary \ref{cor:deltas}.
	Since $P_t(\omega,\dom(\Theta_{t+1}))=1$ for any $\omega\in (N'_t)^C$, we can extend $\delta_{\theta_{t+1}}$ arbitrarily on the complement of $\dom(\Theta_{t+1})$ and, with a slight abuse of notation, we still denote it by $\delta_{\theta_{t+1}}$.
	The product measure $P_{t}\otimes\delta_{\theta_{t+1}}$ belongs to $\hat{\cP}_{t}(\omega)$ for any $\omega\in (N'_t)^C$. 
	This shows $(N'_t)^C\subset (N_t)^C$ and the thesis follows.
	 \end{proof}
\begin{corollary} \label{cor:supports}
	Assume $\NAP$.
	For any $t\in\cI_{-1}$,  $\hat{S}_{t+1}\in\inter(\tilde{K}_{t+1}^*)$ $\hat{\cP}$-q.s. and for any
	$(\omega,\theta)\in N_t^C\times\R^{d-1}$, 
	
	\[
	\supp_{\hat{\cP_t}(\omega)}\hat{S}_{t+1}(\omega,\theta;\cdot)=\supp_{\cP_t(\omega)}\tilde{K}^{*,0}_{t+1}(\omega;\cdot).
	\]
\end{corollary}
\begin{proof}
	It follows from $P_{t}\otimes\delta_{\theta_{t+1}}$ belonging to $\hat{\cP}_{t}$ for any $\delta_{\theta_{t+1}}\in\cL^0(\cF^u_{t+1};\delta_{\Theta_{t+1}})$ and $P_t$ being a measurable selector of $B_t$ in \eqref{eq:Bt}.
	 \end{proof}
Corollary \ref{cor:supports} shows that the role of the parameter $\theta$ for the price process $\hat{S}$ is to ``span'' the dual cones given by the backward recursion \eqref{def:tildeK}.
We set 
\[
\hat{\cP}:=\{\hat{P}_{-1}\otimes\cdots\otimes \hat{P}_{T-1}: \hat{P}_t\in\cL^0(\cF^u_t;\hat{\cP}_t),  \forall t\in\cI_{-1}\}.\]
The class is well defined and constructed via Fubini's Theorem as done for $\cP$.
Indeed, from Lemma \ref{lem:non-empty}, the set $N_t$ is $\cP$-polar, thus, we can extend arbitrarily any $\hat{P}_t$ to a universally measurable kernel which, with a slight abuse of notation, we still denote by $\hat{P}_t$.
By construction, we have that the probability of the set of trajectories taking values in the interior of $\tilde{K}^{*}_t$ is equal to $1$, i.e.,
\begin{equation}\label{eq:hatSinKtilde}
\hat{\P}\left((\omega,\theta)\in \hat{\Omega}: \hat S_t(\omega,\theta)\in\inter\big(\tilde{K}^{*}_t(\omega)\big),\ \forall t\in\cI\right)=1,\ \forall\ \hat{\P}\in\hat{\cP}.
\end{equation}
We finally show that starting from a model $(K,\cP)$ satisfying $\NAP$, the induced frictionless market $(\hat{S},\hat{\cP})$ satisfies the no arbitrage condition of Definition \ref{def:NoA} below.
\begin{definition}\label{def:randomized_strategies}
	We say that a process $H$ is an \emph{admissible strategy} if $H_{t+1}\in \cL^0(\hat{\cF}^u_t;C_t)$  and the \emph{self-financing} condition $(H_{t+1}-H_t)\cdot \hat{S}_t=0$ $\hat{\cP}$-q.s. holds for $0\le t\le T-1$.
	The class of admissible strategies is denoted by $\admr$. 
\end{definition}
\begin{definition}[$\NA$]\label{def:NoA} The \emph{No Arbitrage condition} holds if $(H\circ \hat{S})_T\ge 0$ $\hat{\cP}$-q.s. implies $(H\circ \hat{S})_T= 0$ $\hat{\cP}$-q.s. for any $H\in\admr$.
\end{definition}
In order to use the frictionless duality results of \cite{BZ17} we need to verify Assumption 3.1 and 5.1 in that paper.
Note that the set $\cH_t$, in the notation of \cite{BZ17}, corresponds to the set of constraints $C_t$ considered here. 
Under $\NAP$, Corollary \ref{cor:supports} and Proposition \ref{prop:nonempty}, imply that\[
\lin\big( \supp_{\hat{\cP_t}(\omega)}(\hat{S}_{t+1}(\omega;\cdot)-\hat{S}_t(\omega))\big)=\R^{d-1}\times\{0\},\quad \hat{\cP}\text{-q.s.},
\]
where, for a set $U\subset\R^d$, $\lin(U)$ denotes its linear hull.
We deduce that the sets $\cH_t$, $\cH_t(\hat{\cP})$ and $\cC_{\cH_t}(\hat{\cP})$ in \cite{BZ17}, they all coincide $\hat{\cP}$-q.s. with the first $d-1$ components of the set $C_t$.
Since $C_t$ is a convex closed cone, Assumption 3.1 i)-ii) and 5.1 i) are met.
By \cite[Remark 5.2]{BZ17} it is sufficient to verify Assumption 5.1 ii). 
In particular we show that 
\[
A_t(\hat{\omega},\hat{\P}):= \sup_{x\in C_t(\hat{\omega})}x\cdot \E_{\hat{\P}}[\Delta \hat{S}_{t}(\hat{\omega};\cdot)],\quad \hat{\omega}\in\hat{\Omega}_t,\ \hat{\P}\in\mathfrak{P}(\hat{\Omega}_1),
\]
is Borel-measurable.
To see this observe that $D:=\{(\hat{\omega},\hat{\P}): \E_{\hat{\P}}|\Delta\hat{S}_t(\hat{\omega};\cdot)|<\infty\}$ is Borel-measurable as $\hat{S}$ is Borel-measurable (see, e.g., the proof of \cite[Lemma 4.10]{BN13}).
Moreover, the function $F((\hat{\omega},\hat{\P}),x):=x\cdot \E_{\hat{\P}}[\Delta \hat{S}_{t}(\hat{\omega};\cdot)]$ is a Charath\'eodory map, namely, it is continuous in $x$ when $(\hat{\omega},\hat{\P})$ are fixed and it is measurable in $(\hat{\omega},\hat{\P})$ when $x$ is fixed.
From \cite[Example 14.15]{R}, the random set $F((\hat{\omega},\hat{\P}),C_t(\hat{\omega}))$ is again Borel-measurable.
Finally, $A_t$ restricted to $D$ is again Borel-measurable since, for any $c\in\R$,
\[
A_t^{-1}((c,\infty])\cap D=\{(\hat{\omega},\hat{\P}): F((\hat{\omega},\hat{\P}),C_t(\hat{\omega}))\cap (c,\infty)\neq\emptyset \}\cap D.
\]
Let
$\hat{\cQ}:=\{\Q\ll\hat{\cP}: H\circ\hat{S}\text{ is a }\Q\text{-local-supermartingale }\forall H\in\cL(\hat{\F}^u_-;C)\}$.
The following is Theorem 3.2 of \cite{BZ17} which is also valid in our context.
For $\omega\in\Omega_t$ fixed, $\NAt$ corresponds to $\NA$ for the one period market $(\hat{S}_t(\omega),\hat{S}_{t+1}(\omega;\cdot))$.
\begin{theorem}\label{lem:NA_onestep} The following are equivalent:
	\begin{enumerate}
		\item $\NA$;
		\item for any $0\le t\le T-1$, $N'_t:=\{(\omega,\theta)\in\hat{\Omega}_t: \NAt\text{ fails}\}\in\cF^u$ is a $\hat{\cP}$-polar set;
		\item for any $\P\in\hat{\cP}$ there exists $\Q\in\hat{\cQ}$ such that $\P\ll\Q$.
	\end{enumerate}
\end{theorem}
\begin{proof}
	The only difference from the proof of \cite[Theorem 3.2]{BZ17} is that $\hat{\cP}_t(\omega)$ might have empty values on the $\cP$-polar set $N_t\in\cF^u$.
	Recall $\gr \hat{\cP}_t$ is analytic by Proposition \ref{prop:construction}.
	Thus, also $\dom(\hat{\cP}_t)$ is an analytic set.
	$1.\Rightarrow 2.$ is proven in \cite[Lemma 3.3]{BZ17}.
	It is shown that $(N'_t)^C$ is equal to the set $\{\omega\in\Omega_t: (\Lambda^*\cap C_t)(\omega)\subset-\Lambda^*(\omega)\}$, where we denote $\Lambda(\omega)=\supp_{\hat{\cP_t}(\omega)}(\hat{S}_{t+1}(\omega;\cdot)-\hat{S}_t(\omega))$.
	In our framework, the above set has to be intersected with $\dom(\hat{\cP}_t)$ which is analytic and, therefore, the intersection is again universally measurable.
	The same proof yields that $N'_t$ is $\cP$-polar. 
	
	$2.\Rightarrow 3.$ is based on \cite[Lemma 3.4]{BZ17}.
	The universally measurable kernels $Q_t$ defining $\Q\in\hat{\cQ}$ are constructed outside a $\cP$-polar set and, in particular, they are chosen as selectors of a set $\Xi$ with $\dom(\Xi)=(N'_t)^C$. 
	In our framework, the same $\Xi$ satisfies $\dom(\Xi)=(N_t\cup N'_t)^C$, which is still universally measurable and $\cP$-polar. The same proof allows to conclude.
	
	$3.\Rightarrow 1.$ is standard.
	 \end{proof}

\begin{proposition}\label{prop: RNA implies NA}
	$\NAP$ implies $\NA$.
\end{proposition}
\begin{proof}
	Fix $t\in\{0,\ldots,T-1\}$.
	From Theorem \ref{lem:NA_onestep}, we only need to show that $\NAP$ implies that the set $N'_t$ is $\hat{\cP}$-polar. 
	Suppose that there exists $H\in\cL^0(\hat{\cF}^u_{t};C_t)$, such that 
	\[
	H(\hat{\omega})\cdot\big(\hat S_{t+1}(\hat{\omega};\cdot)-\hat S_{t}(\hat{\omega})\big)\ge 0,\qquad \hat{\cP}_t(\omega)\textrm{-q.s.}  
	\]
	Corollary \ref{cor:supports} implies that $H(\hat{\omega})$ (weakly) separates $\{\hat{S}_{t}(\hat{\omega})\}$ from the set $\supp_{\cP_t(\omega)}\tilde{K}^*_{t+1}(\omega;\cdot)$ for any $\hat{\omega}$ in the complementary of a $\hat{\cP}$-polar set.
	Such a separation extends to the closed convex hull $A_t(\hat{\omega}):=\overline{\conv(\Gamma_t(\omega))}$, with the notation of \eqref{def:tildeK}.
	We can thus rewrite $H(\hat{\omega})\in \big(A_t-\hat{S}_{t})^*(\hat{\omega})$.
	Moreover, from the condition $H(\hat{\omega})\in C_t(\hat{\omega})=C_t^{**}(\hat{\omega})$, we also have $H(\hat{\omega})\in \big(A_t+C_t^*-\hat{S}_{t})^*(\hat{\omega})$.
	Finally, \eqref{def:tildeK} implies that $\inter(\tilde{K}^*_t)\subset \inter(A_t+C_t^*)$.
	We deduce that \[
	H(\hat{\omega})\neq 0\Rightarrow\hat{S}_{t}(\hat{\omega})\notin\inter(\tilde{K}^*_t(\omega)).
	\]
	By Corollary  \ref{cor:supports}, $\hat{S}_{t}\in\inter(\tilde{K}^*_t)$ $\hat{\cP}$-q.s., thus,
	$\{\hat{\omega}\in\hat{\Omega}_t: H(\hat{\omega})\neq 0\}$ is $\hat{\cP}$-polar. 
	 \end{proof}

\section{The Superhedging duality}\label{sec:mainproof}

This section is devoted to the proof of Theorem \ref{thm:main}.
To this aim we compare both the primal and the dual problem with the randomized counterpart in the frictionless market induced by $\hat{S}$ and constructed in Section \ref{sec:Random}.
Using duality results known for the frictionless case we obtain the result.

\paragraph{Equality of the primal problems.}
We first observe that using admissible strategies with respect to $\K$ or with respect to $\tilde{\K}$ yields the same superhedging price.
\begin{lemma}\label{lem:equality admissibility}
	$\sh_{\K}(G)=\sh_{\tilde \K}(G)$.
\end{lemma}
\begin{proof}
	Since $K_t\subset\tilde{K}_t$, the inequality $(\ge)$ is trivial.
	Let now $(y,\tilde{\eta})\in \R\times\cH^{\tilde{K}}$ be a superhedge for $G$. 
	We show that there exists $\eta\in\cH^K$ such that $\eta_T=\tilde{\eta}_T$ and, thus, $(y,\eta)$ is a superhedge for $G$.
	By definition, we can write $\tilde{\eta}_T=\sum_{t=0}^T-\tilde{k}_t$ for some $\tilde{k}_t\in\cL^0(\cF^u_t;\tilde{K}_t)$, for any $t\in\cI$.
	We observe that, from \eqref{def:tildeK},
	\[
	\tilde{K}_t=\tilde{K}^{**}_t=\big(K_t^*\cap\big(\overline{\conv (\Gamma_t)}+C^*_t\big)\big)^*=K_t+(\Gamma_t^*\cap C_t).
	\]
	From \cite[Lemma 8]{BayraktarZhangMOR}, $\tilde{k}_t=f+g$ with $f\in\cL^0(\cF^u_t;K_t)$ and $g\in\cL^0(\cF^u_t;\tilde{K}_{t+1}\cap C_t)$.
	Iterating the same procedure up to time $T-1$ and recalling that $\tilde{K}_T=K_T$ we obtain that
	\[
	\tilde{k}_t=f_t^t+\ldots f_T^t,\quad\text{ for some } f_s^t\in \cL^0(\cF^u_t;K_{s}),\  \forall s\in\{t,\ldots T\}. 
	\] 
	Moreover, $g_s^t:=\sum_{u=s+1}^Tf^t_u$ belongs to $\cL^0(\cF^u_t;\tilde{K}_{s+1}\cap C_s)$ for $s=t,\ldots,T-1$.
	Note that $f_s^t$ is defined only for $s\ge t$. 
	We set $f_s^t=0$ for $s< t$, so that we can rewrite $\tilde{k}_t=\sum_{s=0}^T f_s^t$.
	
	Define now $k_t:=\sum_{s=0}^t f_t^s$ and $\eta_t:=\sum_{u=0}^t-k_u$ for $t\in\cI$.
	Clearly $\eta_T=\tilde{\eta}_T$ so that
	\[
	y\unit+\tilde{\eta}_T-G\in K_T\quad \Rightarrow\quad y\unit+\eta_{T}-G\in K_T.
	\]
	We are only left to show that $\eta_t\in \cL^0(\cF^u_t;C_t)$ for any $t\in\cI$.
	To this aim observe that for $t=T$ it follows from $\eta_T=\tilde{\eta}_T$. 
	For $t=0,\ldots,T-1$, we have
	\[
	\sum_{u=0}^t\tilde{k}_u=\sum_{u=0}^t \left(\sum_{s=0}^t f_s^u+\sum_{s=t+1}^T f_s^u\right)= \sum_{s=0}^t (k_s+g_t^s)
	\]
	where, for the second equality we exchanged the order of summation in the first term and used the definition of $g_t^s$ in the second term.
	The above equation reads as $\eta_t=\tilde{\eta}_t+\sum_{s=0}^t g_t^s$.
	By construction, $g_t^s\in C_t$ $\cP$-q.s.
	Moreover, the admissibility of $\tilde{\eta}$ implies that $\tilde{\eta}_t\in C_t$ $\cP$-q.s. 
	By recalling that $C_t$ is a convex cone, the thesis follows.
	 \end{proof}
We now consider the superhedging problem in the frictionless market defined by $\hat{S}$.
Note that a trading strategy in $\admr$ (see Definition \ref{def:randomized_strategies}) could, in principle, depend on the variable $\theta$.
As this variable is only fictitious, a generic $\hat{\F}^u$-predictable process cannot consistently identify an element in $\adm$. 
To this aim we need to reduce the class of admissible strategies to those which only depend on the variable $\omega$.
\begin{definition}\label{def:strategies}
	A \emph{consistent} strategy $H=\{H_t\}_{t\in\cI}$ is an $\R^d$-valued process satisfying $H_{t+1}\in \cL^0(\cF^u_t\otimes\{\emptyset,\R^d\};C_t)$ for any $0\le t\le T-1$ and the following  \emph{self-financing} condition holds:
	\begin{equation}\label{eq:SF}
	\sup_{\theta\in\Theta_t}(H_{t+1}-H_t)\cdot \hat{S}_t(\cdot,\theta)=0\quad \cP\text{-q.s.}\quad \forall 0\le t\le T-1.
	\end{equation}
	We denote by $\admc$ the set of all self-financing consistent strategies.
\end{definition}
\begin{remark}Recall that the last component of $\hat{S}$ serves as a num\'eraire.
	The self financing condition for $H\in\admr$ is standard, namely, it requires that $-(H^d_{t+1}-H^d_t)$ coincides with $\sum_{i=1}^{d-1}(H^i_{t+1}-H^i_t)\hat{S}^i_t$ $\hat{\cP}$-q.s. On the other hand, a consistent strategy depends only on the $\omega$ variable and hence the position in the num\'eraire needs to be able to cover the worst case scenario for the price of $\hat{S}_t$, which explains \eqref{eq:SF}.
	We show below that, for any consistent strategy, the left hand side of \eqref{eq:SF} is measurable.
\end{remark}
Depending on the choice of the admissible strategies, two corresponding superhedging prices of a random variable $g$ can be computed in the enlarged market:
\begin{eqnarray}
\shr(g)&:=&\inf\big\{y\in\R:\exists H\in\admr\text{ such that }\ y+(H\circ\hat{S})_T\ge g,\quad \hat{\cP}\text{-q.s.}\big\},\label{eq: SH random}\\
\shc(g)&:=&\inf\big\{y\in\R:\exists H\in\admc\,\ \text{ such that }\ y+(H\circ\hat{S})_T\ge g,\quad \hat{\cP}\text{-q.s.}\big\}.\label{eq: SH consistent}
\end{eqnarray}
We want to show that the superhedging price of $G$ is equal to the superhedging price of $G\cdot\hat{S}_T$ in the frictionless market, using only consistent strategies.
Towards this aim, let us first elaborate on the self-financing condition for consistent strategies.
For any $0\le t\le T-1$, let $\Delta H_t:=H_{t+1}-H_t$ and define $F(\omega,x):=\sum_{i=1}^{d-1}\Delta H^i_{t}(\omega)x^i$, which is a Charath\'eodory map.
Recall that the set $\tilde{K}^{*,d-1}_t$ from Step 1 in Lemma \ref{lem:construction} has analytic graph and, thus, it is universally measurable (see e.g. \cite[Lemma 12]{BayraktarZhangMOR}).
From \cite[Example 14.15]{R}, the random set $F(\omega,\tilde{K}^{*,d-1}_t)$ is again $\cF^u_t$-measurable.
We define $\phi_t(\omega):=\sup F(\omega,\tilde K^{*,d-1}_t(\omega))$.
Observe that $\phi_t$ is $\cF^u_t$-measurable, indeed, for any $c\in\R$,
\[
\phi_t^{-1}((c,\infty])=\{\omega\in\Omega_t: F(\omega,\tilde K^{*,d-1}_t(\omega))\cap (c,\infty)\neq\emptyset \}\in\cF^u_t
\]
and $\phi_t^{-1}(\{\infty\})=\cap_{c\in\Q}\phi_t^{-1}((c,\infty])\in\cF^u_t$.
We also observe that the self-financing condition \eqref{eq:SF} can be rewritten as follows.  
\begin{equation}\label{eq:SF_short}
-\Delta H^d_t=\sup_{\theta\in\Theta_t}\sum_{i=1}^{d-1}\Delta H^i_t\ \hat{S}^i_t(\cdot,\theta)=\sup_{x\in \tilde K^{*,d-1}_t}\sum_{i=1}^{d-1}\Delta H^i_t\ x^i=\phi_t,
\end{equation}
where the second equality follows from Corollary \ref{cor:supports}.
Define $A^{\infty}:=\cup_{t=0}^{T-1}\phi_t^{-1}(\{\infty\})$.

\begin{lemma}\label{almost attainment}
	Let $\P\in\cP$, $H\in\admc$.
	\begin{itemize}
		\item Suppose that $\P(A^{\infty})>0$.
		Then, for any $\forall n\in\N$, there exists $\hat{\P}^n\in\hat{\cP}$ such that $\hat{\P}^n_{\mid_\Omega}=\P$ and 
		$\hat{\P}^n(\sum_{i=1}^{d-1}\Delta H^i_t\ \hat S^i_t\ge n)>0$ for some $0\le t\le T-1$;
		\item Suppose that $\P(A^{\infty})=0$.
		Then, for any $n\in\N$, there exists $\hat{\P}^n\in\hat{\cP}$ such that $\hat{\P}^n_{\mid_\Omega}=\P$ and 
		$\sum_{i=1}^{d-1}\Delta H^i_t\ \hat S^i_t\ge -\Delta H^d_t-\frac{1}{n}$, $\hat{\P}^n$-a.s. for any $0\le t\le T-1$.
	\end{itemize}
\end{lemma}

\begin{proof}
	For a fixed $\P$, we might take a Borel-measurable version of $H$ (and therefore of $\phi_t$).
	We define the Charath\'eodory map $\hat{F}(\omega,\theta):=F(\omega,\hat S_t(\omega,\theta))$ with $0\le t\le T-1$ and $F$ as above, but with the difference that now $F$ is Borel-measurable in $\omega$.
	From \cite[Example 14.15 b)]{R}, the random sets
	\[
	\omega\rightarrow\{\theta\times \R^{d-1}: \hat{F}(\omega,\theta)\ge n\}\quad 
	\omega\rightarrow\{\theta\in \R^{d-1}: \hat{F}(\omega,\theta)\ge \phi_t(\omega)-\frac{1}{n}\},
	\]
	are Borel-measurable and, therefore, they have Borel-measurable graph.
	By intersecting their graphs with $\gr(\Theta_t)$ we obtain two analytic sets.
	The Jankov-Von Neumann Theorem provides the existence of universally measurable selectors $\bar\theta^{\infty}_t$ and $\bar\theta^{<\infty}_t$ respectively.
	
	If $\P(A_t^{\infty})>0$ for some $0\le t\le T-1$ we let $\bar\theta_t=\bar\theta^{\infty}_t$ and, for $s\neq t$, we let $\bar\theta_s$ be an arbitrary selector of $\Theta_s$.
	If $\P(A^{\infty})=0$ we let $\bar{\theta}_t=\bar\theta^{<\infty}_t$ for any $0\le t\le T-1$ and $\bar\theta_T$ an arbitrary selector of $\Theta_T$.
	Since $\P$ is fixed, we take Borel measurable versions of the above selectors.
	In both cases, recalling that the map $\theta\rightarrow\delta_\theta$ is an embedding of $\R^d$ into $\mathfrak{P}(\R^d)$, we construct a probability measure $\hat{\P}^n$ from the collection of kernels
	\[\hat{P}_{t}(\omega_0,\ldots,\omega_t;\ \omega',\theta'):=P_t(\omega_0,\ldots,\omega_t;\ d\omega')\otimes \delta_{\bar\theta_{t+1}(\omega_0,\ldots,\omega_t;\ \omega')}(\theta'),
	\]
	for $0\le t\le T-1$ and extend it to $\hat{P}_{-1}:=P_{-1}\otimes \delta_{\bar\theta_{0}}$ for an arbitrary $P_{-1}\in\mathfrak{P}(\Omega_1)$.
	The constructed $\hat{\P}^n$ satisfies $\hat{\P}^n_{\mid_\Omega}=\P$ and the desired properties. 
	 \end{proof}

\begin{lemma}\label{lem: eqSH}
	$\sh_{\tilde \K}(G)=\shc(G\cdot\hat{S}_T)$.
\end{lemma}
\begin{proof}
	\emph{The inequality} $(\ge)$.
	If the set of superhedging strategies for $G$ is empty, then $\sh_{\tilde \K}(G)=\infty$ and the inequality holds trivially.
	Suppose that $(y,\eta)\in \R\times\cH^{\tilde{K}}$ is a superhedge for $G$.
	Define $H_{t+1}:=\eta_t$ for $t=0,\ldots, T-1$.
	Since $\tilde{K}_T$ is a convex cone and $\eta_T-\eta_{T-1}\in -\tilde{K}_T$ by admissibility, we have that 
	\[
	ye_d+\eta_{T}- G\in \tilde{K}_T\quad \Rightarrow \quad
	ye_d+\eta_{T-1}- G\in \tilde{K}_T.
	\]
	For any $\omega$ outside a $\cP$-polar set and for any $s_t\in \tilde{K}^{*,0}_t(\omega)$,
	\begin{eqnarray*}
		0&\le & \left(ye_d+\eta_{T-1}(\omega)- G(\omega)\right)\cdot s_T \\
		&\le&y+H_T(\omega)\cdot s_T+\sum_{t=0}^{T-1}k_t(\omega)\cdot s_t-G(\omega)\cdot s_T\\
		&=&y+ \sum_{t=0}^{T-1}H_{t+1}(\omega)\cdot(s_{t+1}-s_t)-G(\omega)\cdot s_T.\\
	\end{eqnarray*}
	where the second inequality follows from $k_t\in \tilde{K}_t$ for any $t\in\cI$, $\cP$-q.s. 
	Recalling that, from \eqref{eq:hatSinKtilde} we have $\hat{S}_t\in \tilde{K}^*_t$ $\hat{\cP}$-q.s., it follows $y+(H\circ\hat{S})_T\ge G\cdot\hat{S}_T$ $\hat{\cP}$-q.s.
	It remains to show that, without loss of generality, $\eta$ can be chosen such that $H$ is admissible.
	From $\eta\in\cH^{\tilde{K}}$ we have that $\Delta H_t=\eta_t-\eta_{t-1}\in -\tilde{K}_t=(-\tilde{K}^*_t)^*$.
	In particular, this implies $\Delta H_t\cdot \hat{S}_t\le 0$ $\hat{\cP}$-q.s. and, therefore,
	\[
	\delta_t:=\Delta H^d_t+\sup_{\theta\in\Theta_t}\sum_{i=1}^{d-1}\Delta H^i_t\ \hat{S}^i_t(\cdot,\theta)\le 0\quad \cP\text{-q.s.}	
	\]
	Consider the new strategy $\tilde{\eta}$ with $\tilde{\eta_t}=\eta_t-e_d\sum_{u=0}^t\delta_u$ for $t\in\cI$.
	Since $\Delta\tilde{H}^i_t=\Delta H^i_t$ for $i=1,\ldots d-1$ and $-\Delta \tilde{H}^d_t=-\Delta H^d_t+\delta_t$, the self-financing condition follows from \eqref{eq:SF_short}.
	Moreover, using again \eqref{eq:SF_short}, we also have $(\tilde{\eta}_t-\tilde{\eta}_{t-1})\cdot x\le 0$ for any $x\in \tilde K^{*,0}_t$, which implies $\tilde{\eta}_t-\tilde{\eta}_{t-1}\in-\tilde{K}_t$.
	Finally, $\tilde{\eta}_t\in \cL^0(\cF^u_t;C_t)$  since it coincides with $\eta_t$ on the first $d-1$ coordinates and the last one is unconstrained (see Remark \ref{rmk:unconstrained}).
	We conclude by observing that $ye_d+\tilde{\eta}_{T-1}- G=ye_d+\eta_{T-1}-e_d\sum_{t=0}^{T-1}\delta_t- G$
	which belong to $\tilde{K}_T$ since $\sum_{t=0}^{T-1}\delta_t\le 0$ $\cP$-q.s. and $\tilde{K}_T$ is a convex cone containing $\R^d_+$.
	
	\emph{The inequality} $(\le)$.
	If the set of superhedging strategies for $G\cdot\hat{S}_T$ is empty, then $\shc(G\cdot\hat{S}_T)=\infty$ and the inequality holds trivially.
	Suppose that $(y,H)\in \R\times\admc$ is a superhedge for $G\cdot\hat{S}_T$.
	We set $k_T:=0$, $k_{t}:=H_{t}-H_{t+1}$ and $\eta_t:=\sum_{u=0}^t-k_u$ for any $t\in\cI$.
	Suppose first $\P(k^d_t=\infty)>0$ for some $\P\in\cP$ and $t\in\cI$.
	From Lemma \ref{almost attainment}, for any $n\in\N$, there exists $\hat{\P}^n\in\hat{\cP}$ such that
	\[
	-(H^{d}_{t+1}-H^{d}_t)(\omega)\ge  \sum_{i=1}^{d-1}-k^i_t \ \hat S^i_t\ge n,\quad \hat{\P}^n\text{-a.s.}
	\]
	where the first inequality follows from \eqref{eq:SF}.
	Note that the first term on the left is the cost of rebalancing, at time $t$, the strategy $H$ which superhedge $G\cdot\hat{S}_T$ under $\hat{\P}^n$
	Since $n$ is arbitrary and $\hat{\P}^n\in\hat{\cP}$ for any $n\in\N$, we deduce that $\shc(G\cdot\hat{S}_T)=\infty$ and the inequality holds trivially.
	
	For the rest of the proof we suppose that $k^d_t$ is pointwise finite for any $t\in\cI$ (indeed, the $\cP$-q.s. version $k^d_t\mathbf{1}_{\{k^d_t<\infty\}}$ is again universally measurable).
	From \eqref{eq:SF_short},  $k_t(\omega)\cdot x\ge 0$ for any $x\in \tilde K^{*,0}_t$, which implies $k_t(\omega)\in \tilde K_t(\omega)$.
	We now rewrite the superhedging property of $(y,H)$ in terms of $(y,\eta)$.
	For any $(\omega,\theta)$ outside a $\hat{\cP}$-polar set, we have
	\begin{eqnarray}\label{eq:SHrewrite}
	0&\le & y+ \sum_{t=0}^{T-1}H_{t+1}(\omega)\cdot(\hat{S}_{t+1}-\hat{S_t})(\omega,\theta)-G(\omega)\cdot\hat{S}_T(\omega,\theta)\nonumber\\
	&=&y+\sum_{t=0}^{T-1}-k_t(\omega)\cdot(\hat{S}_T-\hat{S}_t)(\omega,\theta)-G(\omega)\cdot\hat{S}_T(\omega,\theta)\nonumber\\
	&=&\left(ye_d+\eta_{T}(\omega)- G(\omega)\right)\cdot\hat{S}_T(\omega,\theta)+\sum_{t=0}^{T-1}k_t(\omega)\cdot\hat{S}_t(\omega,\theta).
	\end{eqnarray}
	We claim that this implies:
	\begin{equation}\label{eq:SHnoterm}
	0\le\left(ye_d+\eta_{T}(\omega)- G(\omega)\right)\cdot\hat{S}_T(\omega,\theta),\quad\hat{\cP}\text{-q.s.}
	\end{equation}
	To prove the claim observe that the first term in \eqref{eq:SHrewrite} depends on $\theta$ only through the last component $\theta_T$, whereas, the second term in \eqref{eq:SHrewrite} depends only on the first $T-1$ components of $\theta$.
	Fix $n\in\N$, $\P\in\cP$. 
	From Lemma \ref{almost attainment} and \eqref{eq:SF_short}, there exists $\hat{\P}^n\in\hat{\cP}$ such that $\hat{\P}^n_{\mid_\Omega}=\P$ and 
	\[
	k_t\cdot \hat{S}_t =\sum_{i=1}^{d-1}k^i_t\hat{S}^i_t+k^d_t\le\frac{1}{n},\quad \hat{\P}^n\text{-a.s.}
	\]
	Since $n\in\N$ and $\P\in\cP$ are arbitrary, we deduce that \eqref{eq:SHnoterm} holds and the claim is proven.
	It remains to show that for $\xi=ye_d+\eta_{T}- G$, 
	\[
	\xi\cdot \hat{S}_T\ge 0\quad \hat{\cP}\text{-q.s.}\quad \Rightarrow\quad  \xi\in \tilde{K}_T\quad \cP\text{-q.s.}
	\]
	Suppose that, by contradiction, there exists a set $A$ and a probability $\P\in\cP$ such that $\P(A)>0$ and $\xi(\omega)\notin \tilde{K}_T(\omega)$ for any $\omega\in A$.
	Without loss of generality, we may take a Borel measurable version of $\xi$.
	Recall that $\tilde{K}_T=K_T$ is assumed to be Borel measurable, so that
	\[
	B:=\{(\omega,y)\in\Omega\times\R^d: \xi(\omega)\cdot y < 0 \}\cap \gr (\inter(\tilde{K}^{*}_T))
	\]
	is Borel measurable from \cite[Lemma A.1]{BN16}.
	Moreover, its projection on $\Omega$ contains $A$.
	Since $B$ is Borel, from Jankov-Von Neumann Theorem, there exists a universally measurable map $s_T:\Omega\rightarrow\R^d$ with $\gr s_T\subset B$. 
	Since $\gr s_T\subset \gr(\inter K^*_T)$ we can normalize with respect to the last component and from \cite[Theorem 14.16]{R} there exists $\cF^u_T$-measurable random vector $\bar \theta_T$ satisfying 
	\[
	\xi(\omega)\cdot \hat{S}_T(\omega,\bar \theta_T(\omega)) < 0,\quad \forall\omega\in A.
	\]
	Take also $\bar\theta_t$ a selector of $\Theta_t$, for any $0\le t\le T-1$.
	For any of the above selectors, we take a Borel measurable version.
	Consider now the probability measure $\hat{\P}\in\hat{\cP}$ obtained from the kernels
	\[\hat{P}_{t}(\omega_0,\ldots,\omega_t;\ \omega',\theta'):=P_t(\omega_0,\ldots,\omega_t;\ d\omega')\otimes \delta_{\bar\theta_{t+1}(\omega_0,\ldots,\omega_t;\ \omega')}(\theta'),
	\]
	with $\hat{P}_{-1}:=P_{-1}\otimes \delta_{\bar\theta_{0}}$ for an arbitrary $P_{-1}\in\mathfrak{P}(\Omega_1)$.
	Note that $\hat{\P}\in\hat{\cP}$ and $\hat{\P}_{|\Omega}=\P$.
	By construction,
	\[
	\hat{\P}(\xi\cdot \hat{S}_T< 0)\ge \P(A)>0,
	\]
	which contradicts the hypothesis.
	 \end{proof}

\paragraph{Equality of the dual problems.}
From \cite[Lemma 5.7]{BZ17} any $\hat{\Q}\in\hat{\cQ}$ admits a disintegration $(\hat{Q}_t)_{t=0,\ldots T-1}$ where $\hat{Q}_t$ is a universally measurable selector of 
\begin{eqnarray}\label{eq:supermtg_cond}
\hat{\cQ}_t(\hat{\omega}):=\big\{\hat{\P}\in\mathfrak{P}(\hat{\Omega}_1)&:& \hat{\P}\ll\hat{\cP}_t(\omega),\quad \E_{\hat{\P}}|\Delta\hat{S}_t(\hat{\omega};\cdot)|<\infty\quad \text{and}\nonumber\\
&& \E_{\hat{\P}}[y\cdot\Delta\hat{S}_t(\hat{\omega};\cdot)]\le 0\quad \forall y\in C_t(\omega)\big\},\quad \hat{\omega}\in\hat{\Omega}_t,
\end{eqnarray}
for $t=0,\ldots T-1$.
Analogously, the set of normalized SCPS $\tilde{\cS}^0$ for the market $\tilde{\K}$ is composed of couples $(Z,\Q)$ for which the disintegration $(Q_t)_{t=0,\ldots T-1}$ of $\Q$ satisfies
\begin{equation}\label{eq:SCPScond}
\E_{Q_t}|\Delta Z_t(\omega;\cdot)|<\infty\quad \text{and}\quad \E_{Q_t}[y\cdot \Delta Z_t]\le 0,\ \forall y\in C_t(\omega),
\end{equation} 
for $t=0,\ldots T-1$.
\begin{proposition}\label{prop:eqDual}
	For any random vector $G\in\cF^u$,
	\begin{equation}
	\sup_{(Z,\Q)\in\tilde{\cS}^0}\E_\Q[G\cdot Z_T]=\sup_{\hat{\Q}\in\hat{\cQ}}\E_{\hat{\Q}}[G\cdot \hat{S}_T].
	\end{equation}
\end{proposition}
\begin{proof}
	Suppose that $(Z,\Q)\in\tilde{\cS}^0$. 
	By construction, $\hat{S}_t$ is a Charath\'eodory map.
	From the implicit mapping theorem (see \cite[Theorem 14.16]{R}), there exists an $\F$-adapted process $(\theta_t)_{t\in\cI}$ with $\theta_t:\Omega_t\rightarrow\R^{d-1}$ and such that
	\begin{equation}\label{eq:implicit}
	\hat{S}_t(\omega,\theta_t(\omega))=Z_t(\omega).
	\end{equation}
	Denote by $(Q_t)_{t\in\cI_{-1}}$ the collection of conditional probabilities of $\Q$ given $\cF_{t-1}$, extended to $t=-1$ with an arbitrary $Q_{-1}\in\mathfrak{P}(\Omega_1)$.
	We use $\delta_{\theta_t}$ as a stochastic kernel and construct the probability $\hat{\Q}:= (Q_{-1}\otimes\delta_{\theta_0})\otimes\cdots \otimes(Q_{T-1}\otimes\delta_{\theta_{T}})$.
	Since $(Z,\Q)$ satisfies \eqref{eq:SCPScond} and \eqref{eq:implicit} holds, we deduce that that $\hat{\Q}\in\hat{\cQ}$. 
	Moreover, it is clear that $\E_\Q[G\cdot Z_T]=\E_{\hat{\Q}}[G\cdot \hat{S}_T]$.
	
	Conversely, suppose $\hat{\Q}\in\hat{\cQ}$.
	We define $(Z,\Q)$ as $\Q:=\hat{\Q}_{\mid_\Omega}$ and, for every $t\in\cI$, $Z_t:=\E_{\hat{\Q}}[\hat{S}_t\mid\cF_t]$.
	Denote by $(\hat{Q}_t)_{t\in\cI_{-1}}$ (respectively $(Q_t)_{t\in\cI_{-1}}$) the disintegration of $\hat{\Q}$ (respectively of $\Q$).
	From $\E_{\hat{Q}_t}[y\cdot\Delta\hat{S}_{t+1}(\hat{\omega};\cdot)]\le 0$ for any $y\in C_t$ and from the fact that $C_t$ is $\cF_t$-measurable, we deduce that $Q_t$ satisfies \eqref{eq:SCPScond} for any $t=0,\ldots, T-1$.
	Moreover, since $\hat{\Q}\ll\hat{\cP}$, by definition of $\hat{\cP}$ we obtain: i) $\Q\ll\cP$ and ii) $Z_t$ takes values in $\inter(\tilde{K}^*_t)$ $\Q$-a.s. for any $t\in\cI$. 
	We conclude that $(Z,\Q)\in\tilde{\cS}^0$. Moreover, we obviously have $\E_\Q[G\cdot Z_T]=\E_{\hat{\Q}}[G\cdot \hat{S}_T]$.
	 \end{proof}

\subsection{Proof of Theorem \ref{thm:main}. }

We are now ready to prove the main result of the section. Note that from Lemma \ref{lem:equality admissibility} and \ref{lem: eqSH}, we can only deduce the equality of the primal problems if one restricts to \emph{consistent} trading in the enlarged market (compare with \eqref{eq: SH consistent}). 
It remains to show that the same price is obtained with \emph{randomized} strategies as defined in \eqref{eq: SH random}, in other words, we need to prove that $\shc(G\cdot \hat{S}_T)=\shr(G\cdot\hat{S}_T)$. 
Denote by $USA(\hat{\Omega}_t,t)$ the class of $g:\hat{\Omega}_t\rightarrow\R$ upper semianalytic functions which depends on $\theta$ only through $\theta_t$, i.e.,
\[g(\omega,\theta)=g(\omega',\theta'),\qquad \forall (\omega,\theta),(\omega',\theta')\in\hat{\Omega}_t,\text{ with } \omega=\omega' \text{ and }\theta_t=\theta'_{t}.
\]
\paragraph{The one-period case.} We obtain first the results for $T=1$ which will constitute the building blocks for the general case.
\begin{proposition}\label{prop: eq primals} Suppose $T=1$ and $g\in USA(\hat{\Omega}_T,T)$. If $\NA$ holds true, then  $\shc(g)= \shr(g)$.
\end{proposition}
\begin{proof}
	The inequality $\shc(g)\ge\shr(g)$ is trivial.
	For the converse, let $B_{n}(0)$ be the closed ball in $\R^{d}$ with center in $0$ and radius $n\in\N$.
	The intersection $\tilde{K}^{*,0}_0\cap B_{n}(0)$ is a compact set of the form $O^n\times\{1\}$ for $O^n$ a compact subset of $\R^{d-1}$.
	Recall the definition of $\hat{\cP}_0$ from \eqref{eq:def_random sets} and let $P_{-1}\in\mathfrak{P}(\Omega_1)$ be arbitrary\footnote{Recall that the cone $\tilde{K}^{*,0}_0$ is non-random. Thus, in the enlarged market, the only relevant variable is $\theta$.}. 
	We define
	\[
	\hat{\cP^n}:=\{(P_{-1}\otimes\delta_\theta)\otimes\hat{\P}:\theta\in O^n,\  \hat{\P}\in\hat{\cP}_0\}\subset\hat{\cP}.\]
	
	Denote by $\shc_n$ and $\shr_n$ the analogous of $\shc$ and $\shr$ in equations \eqref{eq: SH random} and \eqref{eq: SH consistent} with $\hat{\cP^n}$ replacing $\hat{\cP}$ and note that, by construction, $\{\shc_n(g)\}_{n}$ and $\{\shr_n(g)\}_{n}$ are increasing sequences bounded from above by $\shc(g)$ and $\shr(g)$ respectively.
	We use now a minimax argument as in \cite{BDT17} to deduce that
	\begin{eqnarray*}
		\shc_n(g)&=&
		\inf_{H\in C_0}\sup_{\theta\in O^n}\sup_{\hat{\P}\ll\hat{\cP}_0}\E_{\hat{\P}}[g-H\cdot(\hat{S}_1-\hat{S_0})]\\
		&=& \sup_{\theta\in O^n}\inf_{H\in C_0}\sup_{\hat{\P}\ll\hat{\cP}_0}\E_{\hat{\P}}[g-H\cdot(\hat{S}_1-\hat{S_0})]\\
		&=& \shr_n(g).
	\end{eqnarray*}
	To justify the above it is sufficient to observe that the function
	\[
	(H,\theta)\mapsto\sup_{\hat{\P}\ll\hat{\cP}_0}\E_{\hat{\P}}[g-H\cdot(\hat{S}_1-\hat{S_0})]=\sup_{\hat{\P}\ll\hat{\cP}_0}\E_{\hat{\P}}[g-H\cdot\hat{S}_1]-H\cdot\hat{S_0}(\omega,\theta),
	\]
	is convex in $H$ for $\theta$ fixed and affine in $\theta$ for $H$ fixed. We can thus apply the minimax theorem of \cite[Corollary 2]{T72}. 
	If $\shr(g)\ge\lim_{n\rightarrow\infty}\shr_n(g)=\infty$, then $\shr(g)\ge \shc(g)$ holds trivially and the proof is complete.
	Suppose the limit is finite. Let $\varepsilon>0$ be arbitrary and, for any $n\in\N$, let $H_n\in C_0$ be an $\varepsilon$-optimal strategy for $\shc_n(g)$, namely,
	\begin{equation}\label{eq:sequence_optimal}
	\shc_n(g)+\varepsilon + H_n\cdot (\hat{S}_1(\omega,\theta_1)-\hat{S_0}(\omega,\theta_0))\ge g(\omega,\theta_1),
	\end{equation}
	for any $(\omega,\theta_1)$ outside a $\hat{\cP}_0$-polar set and for any $\theta_0\in O^n$.
	If $\{H_n\}_{n\in\N}\subset C_0$ is bounded, it admits a convergent subsequence.
	Denote by $\bar H$ its limit. 
	From \eqref{eq:sequence_optimal},  $O^n\times\{1\}\uparrow\tilde{K}^{*,0}_0$ and $\shr(g)\ge \lim_{n\rightarrow\infty}\shr_n(g)=\lim_{n\rightarrow\infty}\shc_n(g)$, it holds
	\[
	\shr(g)+\varepsilon+ \bar H\cdot (\hat{S}_1-\hat{S_0})\ge g,\quad \hat{\cP}\text{-q.s}
	\]
	from which $\shc(g)\le \shr(g)+\varepsilon$.
	We show now that if $\{H_n\}_{n\in\N}\subset C_0$ is unbounded, it contradicts $\NA$.
	This together with $\varepsilon>0$ arbitrary yields the desired inequality.
	To see this divide by $c_n:=\|H_n\|$ both sides of \eqref{eq:sequence_optimal}. 
	Since $\shc_n(g)=\shr_n(g)$ is assumed to be bounded, $\shc_n(g)/c_n$ converges to $0$.
	$H_n/c_n$ belongs to the compact sphere of $\R^d$, thus, there exists $\bar{H}$ with $\|\bar H\|=1$ such that $H_n/c_n\rightarrow\bar{H}$ (up to extracting a subsequence).
	Note that $\bar{H}\in C_0$ since $C_0$ is a closed cone.
	By the same argument as above, this implies, 
	\[
	\bar H\cdot (\hat{S}_1-\hat{S_0})\ge 0,\quad \hat{\cP}\text{-q.s}.
	\]
	Recalling \eqref{eq:hatSinKtilde}, the condition $\NA$ implies $\bar H=0$, which is a contradiction since $\|\bar H\|=1$.
	This concludes the proof.
	 \end{proof}
Note that if $G:\Omega\rightarrow\R^d$ is a Borel-measurable vector, $G\cdot \hat{S}_T:\hat{\Omega}\rightarrow\R$ is a Borel-measurable function which depends on $\theta$ only through $\theta_T$.
In particular Proposition \ref{prop: eq primals}, together with Lemma \ref{lem:equality admissibility} and Lemma \ref{lem: eqSH}, yields $\pi_K(G)= \shr(G\cdot\hat{S}_T)$.

\begin{proof}[\textbf{Proof of Theorem \ref{thm:main} for $\mathbf{T=1}$.}]
	From Proposition \ref{prop: RNA implies NA}, $\NAP$ implies $\NA$ for the enlarged market.
	From Lemma \ref{lem:equality admissibility}, Lemma \ref{lem: eqSH} and Proposition \ref{prop: eq primals}, $\sh_\K(G)=\shr(G\cdot\hat{S}_T)$ which is the superhedging price of $G\cdot\hat{S}_T$, in the enlarged market.
	We show that
	\[
	\sh_\K(G)=\shr(G\cdot\hat{S}_T)=\sup_{\hat{\Q}\in\hat{\cQ}}\E_{\hat{\Q}}[G\cdot \hat{S}_T]=\sup_{(Z,\Q)\in\tilde{\cS}^0}\E_\Q[G\cdot Z_T]\le\sup_{(Z,\Q)\in\cS^0}\E_\Q[G\cdot Z_T].
	\]
	Indeed, the second equality follows from \cite[Theorem 4.3]{BZ17} after observing that, when $C_0$ is a cone, $A^\Q$ in the aforementioned paper is finite if and only if $\hat{\Q}\in\hat{\cQ}$.
	The third equality follows from Proposition \ref{prop:eqDual} and the last inequality follows from $\tilde{\cS}^0\subset \cS^0$.
	The converse inequality follows from standard arguments.
	From \cite[Theorem 4.3]{BZ17} an optimal superhedging strategy exists in the enlarged market, when the price is finite.
	The proofs of Lemma \ref{lem:equality admissibility} and \ref{lem: eqSH} provide the construction of an optimal strategy in the original market.
	 \end{proof}

\paragraph{The multi-period case.}
From \cite[Lemma 3.4]{BZ17}, $\hat{\cQ}_t$ as in \eqref{eq:supermtg_cond} has analytic graph for every $0\le t\le T-1$.
Given $g_{t+1}\in USA(\hat{\Omega}_{t+1},t+1)$, we define
\begin{eqnarray}
g_t:\hat{\Omega}_t\rightarrow \overline{\R} &\text{ as }&g_t(\hat{\omega})=\sup_{\hat{\Q}\in\hat{\cQ}_t(\hat{\omega})}\E_{\hat{\Q}}[g_{t+1}]\label{eq:gt1}\\
g'_t:\Omega_t\times\R^d\rightarrow \overline{\R} &\text{ as }&g'_t(\omega,h)=\sup_{\theta\in\Theta_t(\omega)}\{g_t(\omega,\theta)-h\cdot\hat{S}_t(\omega,\theta)\},\label{eq:gt2}
\end{eqnarray} 
with $\Theta_t$ as in \eqref{def:theta_t}. 
It is possible to show that $g_t\in USA(\hat{\Omega}_t,t)$. Indeed, the measurability property follows exactly from the same argument as in the first lines of the proof of \cite[Lemma 4.10]{BN13}. Moreover, $g_t$ depends on $\theta$ only through $\hat{S}_t$, thus, only through $\theta_t$ (see \eqref{eq:shadow}).
Recall now that the sum of two upper semi-analytic functions is again upper semi-analytic (see e.g. \cite[Lemma 7.30]{BS78}).
Since $\hat{S}_t$ is Borel measurable we deduce that $g_t-h\cdot\hat{S}_t$ is an upper semi-analytic function of $(\omega,h,\theta)$.
From Lemma \ref{lem:construction} and \cite[Proposition 7.47]{BS78} we deduce that $g'_t$ is upper semi-analytic.
Let $\tilde{\cP}_t$ be the set of probabilities on $\Omega_t\times\R^{d-1}\times\hat{\Omega}_1$ given by
\[ \tilde{\cP}_t(\omega):=\{(\delta_{\omega}\otimes\delta_{\theta})\otimes\hat{\P} : (\omega,\theta)\in \gr\Theta_t,\ \hat{\P}\in\hat \cP_{t}(\omega)\},\quad \omega\in\Omega_t.
\]
Recall that the random sets $\hat \cP_{t}$ and $\delta_{\Theta}$ from Corollary \ref{cor:deltas} have analytic graph. 
Since the map $x\mapsto \delta_x$ is an embedding and the map $(P,Q)\mapsto P\otimes Q$ is continuous (see \cite[Lemma 7.12]{BS78}), it follows that also $\tilde{\cP}_t$ has analytic graph.
\begin{lemma}\label{lem:norm int}
	For any $0\le t\le T-1$, the function $f:\Omega_t\times\R^d\times\R^d\rightarrow\overline{\R}$ defined as
	\begin{equation}\label{eq:normal int}
	f(\omega,h,x)=\sup_{\tilde{\P}\ll\tilde{\cP}_t(\omega)}\E_{\tilde{\P}}\big[g_{t+1}-h\cdot\hat{S}_t-x\cdot(\hat{S}_{t+1}-\hat{S}_{t})\big]
	\end{equation}
	is a universally measurable normal integrand.
\end{lemma}
\begin{proof}
	Denote by $f^{\tilde{\P}}(\omega,h,x)$ the functions on the right hand side of \eqref{eq:normal int} for which the supremum is taken.
	From \cite[Corollary 14.41]{R} we need to check:
	\begin{enumerate}
		\item for any $(\omega,h)\in\Omega_t\times\R^d$, the function $f(\omega,h,\cdot)$ is lower semi-continuous.
		\item for any $x\in\R^d$ there exists $\varepsilon'>0$ such that, for all $\varepsilon\in(0,\varepsilon')$, the function
		\[ 
		\Psi_\varepsilon:(\omega,h)\mapsto \inf_{\tilde{x}\in B_\varepsilon(x)} f(\omega,h,\tilde x)
		\]
		is universally measurable, where $B_\varepsilon(x)$ denotes the closed ball of radius $\varepsilon$ centerd in $x$.
	\end{enumerate}
	Since $f^{\tilde{\P}}(\omega,h,\cdot)$ is continuous for every $\tilde{\P}$ and the pointwise supremum of continuous functions is lower semi-continuous, the first claim follows.
	Consider $\varepsilon>0$ arbitrarily.
	We first show that for any $(\omega,h)$,
	\[
	\Psi_\varepsilon(\omega,h)=\sup_{\tilde{\P}\ll\tilde{\cP}_t(\omega)}\inf_{\tilde{x}\in B_\varepsilon(x)} f^{\tilde{\P}}(\omega,h,\tilde x).
	\]
	This follows from the application of a minimax Theorem (see e.g. \cite[Corollary 2]{T72}).
	$B_\varepsilon(x)$ is a compact set and, for fixed $\tilde{x}$, the map $f^{\tilde{\P}}(\omega,h,\tilde x)$ is linear (hence concave) in $\tilde{\P}$.
	On the other hand $\{\tilde{\P}\ll\tilde{\cP}_t(\omega)\}$ is a convex set and, for fixed $\tilde{\P}$, the map $f^{\tilde{\P}}(\omega,h,\tilde x)$ is affine (hence convex) and continuous in $\tilde{x}$.
	We can rewrite $\inf_{\tilde{x}\in B_\varepsilon(x)} f^{\tilde{\P}}(\omega,h,\tilde x)=f_1(\omega,h,\tilde{\P})+f_2(\omega, \tilde{\P})$ where 
	\[
	f_1= \E_{\tilde{\P}}\big[g_{t+1}-h\cdot\hat{S}_t\big],\qquad f_2=-\sup_{\tilde{x}\in B_\varepsilon(x)}\E_{\tilde{\P}}\big[\tilde{x}\cdot(\hat{S}_{t+1}-\hat{S}_{t})\big].
	\]
	$f_1$ is an upper semi-analytic function on $\Omega\times\R^d\times\mathfrak{P}(\Omega)$ (see \cite[Lemma 4.10]{BN13}).
	We claim that $f_2$ is a Borel-measurable function (hence upper semi-analytic).
	Given the claim we observe that $f_1+f_2$ is again upper semi-analytic.
	Moreover, by the same argument for the measurability of \eqref{eq:gt2} above, we can conclude that $\Psi_\varepsilon=\sup_{\tilde{\P}\ll\tilde{\cP}_t(\omega)}(f_1+f_2)$ is upper semi-analytic on $\Omega\times\R^d$ and therefore universally measurable. This proves the second property.
	
	To conclude the proof it is enough to show that $f_2$ is Borel measurable.
	To see this observe that the function $\E_{\P}[\tilde x\cdot(\hat{S}_{t+1}(\omega;\cdot)-\hat{S}_{t}(\omega;\cdot))]$ is measurable in $(\omega,\P)$ and continuous in $\tilde x$, namely, it is a Carath\'eodory map. 
	From \cite[Example 14.15]{R} its composition with $B_\varepsilon(x)$ yields a Borel-measurable random set $A$ such that $\sup A= -f_2$.
	To conclude, observe that for an arbitrary $c\in\R$,
	\[
	\left\{(\omega,\P): f_2<c \right\}= \left\{(\omega,\P): A \cap(-c,\infty)\neq\emptyset \right\},
	\]
	is a Borel set from the measurability of $A$.
	 \end{proof}

\begin{remark}Note that, for any $\omega\in\Omega_t$, the right hand side of \eqref{eq:normal int} is equal to the $\inf\{K\in\R: X\le K\  \hat{\cP}_t(\omega)$-q.s$\}$, where $X$ is the random variable inside the expectation.
	In particular, this is equal to the minimal amount, at time $t$, for which the strategy $x$ is a superhedge for $g_{t+1}$ given that $h$ is the strategy used at time $t-1$.
	Moreover, by construction of $\tilde{\cP}_t$, the strategy $x$ with the initial amount $f(\omega,h,x)$, is a (conditional) superhedging strategy which depends only on the event $\omega$ and not on the event $(\omega,\theta)$.
	In the terminology of Definition \ref{def:strategies}, this construction provides \emph{consistent strategies}. 
\end{remark}
Recall that $\NAt$ is the conditional version of $\NA$ (see Theorem \ref{lem:NA_onestep}).
\begin{proposition}\label{prop:Concatenation}
	Let $0\le t\le T-1$ and assume $\NAt$.
	There exists a universally measurable map $\varphi:\Omega_t\times\R^d\rightarrow\R^d$ and a $\cP$-polar set $N$ such that for any $(\omega,h)\in N^C\times\R^d$ 
	\[
	g'_t(\omega,h)+h\cdot \hat{S}_t+\varphi(\omega,h)\cdot (\hat{S}_{t+1}-\hat{S}_{t})\ge g_{t+1},\ \tilde{\cP}_t\text{-q.s.}
	\]
	and $g'_t(\omega,h)>-\infty$.
\end{proposition}
\begin{proof}
	Define the \emph{consistent} conditional superhedging price of $g_{t+1}$ given $(\omega,h)$ as the map
	\[
	(\omega,h)\mapsto\inf\big\{y\in\R:\exists H\in C_t(\omega)\,:\, y+h\cdot \hat{S}_t+H\cdot (\hat{S}_{t+1}-\hat{S}_{t})\ge g_{t+1},\quad \tilde{\cP}_t(\omega)\text{-q.s.}\big\}.
	\]
	The fact that $g'_t(\omega,h)$ is the consistent conditional superhedging price of $g_{t+1}$, follows from the same minimax argument of Proposition \ref{prop: eq primals} and \cite[Theorem 4.3]{BZ17}.
	Moreover, from Theorem \ref{lem:NA_onestep}, $\NAt$ holds outside a polar set $N$. Again from \cite[Theorem 4.3]{BZ17}, $g'_t(\omega,h)>-\infty$ on $N^C\times\R^d$ and the infimum is attained. It remains to show that a superhedging strategy can be chosen in a measurable way. 
	From Lemma \ref{lem:norm int} the map $f$ defined in \eqref{eq:normal int} is a universally measurable normal integrand.
	From \cite[Proposition 14.33]{R} and recalling that $g'_t$ is upper semi-analytic (hence universally measurable), the map
	\[
	\Psi(\omega,h):=\{x\in\R^d: f(\omega,h,x)\le g'_t(\omega,h)\}
	\]
	is a closed-valued universally measurable random set.
	The desired $\varphi$ is any measurable selector of $\Psi$ (which exists from, e.g. \cite[Corollary 14.6]{R}).
	 \end{proof}

\begin{proof}[\textbf{Proof of Theorem \ref{thm:main} for $\mathbf{T>1}$.}]
	We first show that for any $g\in USA(\hat{\Omega}_T,T)$, 
	\begin{equation}\label{duality}
	\shc(g)=\shr(g)=\sup_{\hat{\Q}\in\hat{\cQ}}\E_{\hat{\Q}}[g].
	\end{equation}
	For $T=1$, \eqref{duality} follows from Proposition \ref{prop: eq primals} and \cite[Theorem 4.3]{BZ17}. We prove the general case by induction. 
	Suppose that for $T=1,\ldots,t$ equation \eqref{duality} is proven.
	Denote by $\shc_{u}$ the (consistent) superhedging price if the terminal time is $u$. In particular, $\shc_T=\shc$ defined in \eqref{eq: SH consistent}.
	Let $g_{t+1}\in USA(\hat{\Omega}_{t+1},t+1)$ and define $g_t$ and $g'_t$ as in \eqref{eq:gt1} and \eqref{eq:gt2} respectively.
	We claim that $\shc_{t+1}(g_{t+1})\le\shc_{t}(g_{t})$.
	Denote by ${\hat{\cP}}_{\mid_t}$ the restriction of $\hat{\cP}$ to $\hat{\Omega}_t$. Similarly for $\hat{\cQ}_{\mid_t}$.
	Consider an arbitrary $(y,H)\in\R\times\admr$ satisfying $y+(H\cdot\hat{S})_t\ge g_{t}$ $\hat{\cP}_{\mid_t}$-q.s.  
	By rewriting the previous inequality we observe that  
	\[
	y+(H\cdot \hat{S})_{t-1}-H_t\cdot\hat{S}_{t-1}\ge g_t-H_t\cdot\hat{S}_t,\quad \hat{\cP}_{\mid_t}\text{-q.s.}
	\]
	which, in turn, implies
	\[
	y+(H\cdot \hat{S})_{t-1}-H_t\cdot\hat{S}_{t-1}\ge g'_t(\cdot,H_t)\quad \hat{\cP}_{\mid_t}\text{-q.s.}
	\]
	Given the strategy $(y,H_1,\ldots,H_t)$, Proposition \ref{prop:Concatenation} provides a universally measurable random vector $H_{t+1}=\varphi(\cdot,H_t)$ such that the strategy $(y,H_1,\ldots,H_t,H_{t+1})$ satisfies $y+(H\cdot\hat{S})_{t+1}\ge g_{t+1}$ $\hat{\cP}$-q.s. 
	From $(y,H)$ above being arbitrary, the claim is proven. We deduce that
	\[
	\shc_{t+1}(g_{t+1})\le\shc_{t}(g_{t})=\shr_{t}(g_{t})=\sup_{\hat{\Q}\in\hat{\cQ}_{\mid_t}}\E_{\hat{\Q}}[g_t]\le \sup_{\hat{\Q}\in\hat{\cQ}}\E_{\hat{\Q}}[g_{t+1}],
	\]
	where the equalities follow from the inductive hypothesis and the second inequality follows from a standard pasting argument. 
	By definition, $\shr_{t+1}(g_{t+1})\le\shc_{t+1}(g_{t+1})$, moreover, the inequality $\sup_{\hat{\Q}\in\hat{\cQ}}\E_{\hat{\Q}}[g_{t+1}]\le \shr_{t+1}(g_{t+1})$ is standard. We conclude that \eqref{duality} holds for $T=t+1$.
	We now choose $g=G\cdot\hat{S}_{T}$, which is Borel-measurable by assumption.
	From Lemma \ref{lem:equality admissibility}, Lemma \ref{lem: eqSH}, Proposition \ref{prop:eqDual} and \eqref{duality}, we deduce
	\[
	\sh_\K(G)= \sup_{(Z,\Q)\in\tilde{\cS}^0}\E_\Q[G\cdot Z_{T}]\le\sup_{(Z,\Q)\in\cS^0}\E_\Q[G\cdot Z_{T}].
	\]
	Again, the converse inequality follows by standard arguments and the duality follows (recall also the discussion before Remark \ref{rmk:unconstrained}).
	Finally, the attainment property in the frictionless market follows from \cite[Theorem 6.1]{BZ17}. 
	The proofs of Lemma \ref{lem:equality admissibility} and Lemma \ref{lem: eqSH} provide the construction of an optimal strategy in the original market.
	 \end{proof}
\subsection{The case with options.}
We now consider the case where a finite number of options $\varphi_1,\ldots,\varphi_e$  are available for semi-static trading.
In this section we show that this case can be embedded in the previous one.
For any $k=1,\ldots,e$, we assume that $\varphi_k:\Omega\to\R^d$ is a Borel-measurable function representing the terminal payoff of an option, in terms of physical units of an underlying $d$-dimensional asset.
Any $\varphi_k$ has bid and ask price at time $0$ denoted, respectively, by $b_k$ and $a_k$.
We set $\Phi:= [\varphi_1;\cdots;\varphi_e;-\varphi_1;\cdots;-\varphi_e]$ with corresponding prices  $p:=(a_1,\ldots,a_e,-b_1,\ldots,-b_e)^T$.
$\Phi$ takes values in $\R^{d\times m}$ and $p\in\R^m$ with $m:=2e$.
For ease of notation we relabel the options and incorporate their price in the payoff so that $\Phi=[\phi_1-p_1\unit;\cdots;\phi_m-p_m\unit]$.
In addition, we suppose that we are given a dynamic trading market $(\K,C)$ satisfying all the hypothesis of Section \ref{sec:main}.
An admissible strategy has the form $\bar{\eta}:=(\eta,\alpha)$ with $\eta\in\adm$ a dynamic strategy and $\alpha\in\R^{m}_+$.
\begin{definition} We say that $\NAPo$ holds if $\NAP$ holds for the dynamic trading market $(\K,C)$ and $\eta_T+\Phi\alpha\in K_T$ $\cP$-q.s implies $\alpha=0$.
\end{definition}
The (semi-static) superhedging price of $G:\Omega\to\R^d$ is defined as
\begin{equation}\label{eq:defSHsemi}
\sh_{\K,\Phi}(G):=\inf\left\{y\in\R:\exists \bar{\eta}\in\adm\text{ s.t. } y\unit+\eta_{T}+\Phi\alpha-G\in K_T,\ \cP\text{-q.s.}\right\},
\end{equation}
where $\unit$ is the $d^{th}$ vector of the canonical basis of $\R^d$.
Finally, we define the set $\cS_{\Phi}:=\{(Z,\Q)\in\cS: \E_\Q[\varphi_k\cdot Z_T]\in(b_k,a_k)\ \forall k=1,\ldots e\}$.
\begin{theorem}\label{thm:options}
	Assume $\NAPo$. For any Borel-measurable random vector $G$,
	\begin{equation}
	\sh_{\K,\Phi}(G)= \sup_{(Z,\Q)\in\cS_{\Phi}}\E_\Q[G\cdot Z_T].
	\end{equation}
	Moreover, the superhedging price is attained when $\sh_{\K,\Phi}(G)<\infty$.
\end{theorem}
As in Section \ref{sec:Random}, we construct an extended space where only dynamic trading in a frictionless asset $\bar{S}$ is allowed.
We set $\bar\Omega_1=\hat{\Omega}_1\times \R^m$ and $\bar{\Omega}=\bar{\Omega}_T$ with $\bar{\Omega}_t$ the $(t+1)$-fold product of $\bar{\Omega}_1$.
We define $\bar{S}_t:\bar\Omega_t\to \R^d\times\R^m$  such that
\begin{itemize}
	\item On the first $d$ components:  $\bar{S}_t(\omega,\theta,x)=\hat{S}_t(\omega,\theta)$;
	\item On the last $m$ components:  $\bar{S}_t(\omega,\theta,x)=x$ for $1\le t\le T-1$, and
	\[
	\bar{S}^k_0=p_k,\quad \bar{S}^k_T(\omega,\theta,x)=\phi_k(\omega)\cdot \hat{S}_T(\omega,\theta),\quad k=d+1,\ldots, m.
	\]
\end{itemize}
The set of priors $\bar{\cP}$ is obtained from the collection $\bar\cP_t:=\hat{\cP}_t\otimes\mathfrak{P}(\R^m)$.
The set of constraints in the frictionless market is obtained as $C\times\R^m_+$.
The set of randomized and consistent strategies in the frictionless market are defined as before and denoted here as $\bar\cH^r$ and $\bar{\cH}$.
Similarly for the corresponding superhedging prices $\bar{\pi}^r$ and $\bar{\pi}$.
We here define the semi-static consistent superhedging price as
\begin{eqnarray*}
	\bar{\pi}_{\Phi}(g):=\inf\big\{y\in\R&:&\exists H\in\bar{\cH}\,\ \text{ such that }y+(H\circ\bar{S})_T\ge g\ \bar{\cP}\text{-q.s. and}\\
	&& H^k_t=H^k_1\text{ for any } k=d+1,\ldots,m,\  t=1,\ldots,T \big\}.
\end{eqnarray*}
We only need to show the following.
\begin{lemma}\label{lem:options}
	$\bar{\pi}(G\cdot\hat{S}_T)=\bar{\pi}_{\Phi}(G\cdot\hat{S}_T)$.
\end{lemma}
\begin{proof}
	The inequality $(\le)$ is clear as any strategy for the right hand side is also allowed for the left hand side.
	For the inequality $(\ge)$, suppose that $(y,\bar{H})\in \R\times\bar{\cH}$ is a superhedge for $G\cdot\hat{S}_T$.
	Let $\bar{H}=(H,h)$, where $H$ is the vector of the first $d$ components and $h$ is the vector of the last $m$ components.
	Recall that $\bar{H}\in\bar{\cH}$ is consistent, i.e, it only depends on the $\omega$ variable.
	Let $A_t:=\cup_{k=d+1}^m\{\omega\in\Omega_t: \bar{H}^k_{t+1}\neq \bar{H}^k_1 \}$ and let $\bar{t}$ be the first time $0\le t\le T-1$ such that $\P(A_t)>0$ for some $\P\in\cP$.
	The superhedging property reads as,
	\begin{eqnarray*}
		y+(\bar{H}\circ \bar{S})_T&=&y+(H\circ \hat{S})_T+\sum_{t=1}^{T-1}\sum_{k=d+1}^m h^k_{t+1}(\bar{S}^k_{t+1}-\bar{S}^k_t)\\
		&=&y+(H\circ \hat{S})_T+\sum_{k=d+1}^m h^k_1 (\bar{S}^k_{\bar{t}}-\bar{S}^k_0)+ \sum_{t=\bar{t}}^{T-1}\sum_{k=d+1}^m h^k_{t+1}(\bar{S}^k_{t+1}-\bar{S}^k_t)\\
		&\ge& G\cdot \hat{S}_T.\quad \bar{\cP}\text{-q.s.}
	\end{eqnarray*}
	Take now $\xi$ a measurable selector of $\{x\in\R^m: x\cdot h_{\bar{t}+1}<0\}$.
	This  exists from \cite[Lemma 12-13]{BayraktarZhangMOR} as the random set corresponds to the interior of the polar cone of $h_{{\bar t}+1}$.
	Since $\P$ is fixed we might take a Borel-measurable version of $\xi$.
	Let $(P_t)_{t=0,\ldots,T-1}$ be the kernel decomposition of $\P$, extended arbitrarily to $t=-1$.
	Fix $x\in\R^m$ and $\delta_{\theta_t}$ an arbitrary selector of $\delta_{\Theta_t}$ from Corollary \ref{cor:deltas}, for any $t\in\cI$.
	For any $\lambda>0$,  define the probability kernels \[\bar P_{\bar{t}}:=P_{\bar{t}}\otimes\delta_{\theta_{\bar{t}+1}}\otimes\delta_{\lambda\xi},\quad\text{and}\quad P_{t}:=P_{t}\otimes\delta_{\theta_{t+1}}\otimes\delta_{x}\ \text{ for } t\neq\bar{t}.\]
	The measures $\bar{\P}^\lambda$ constructed via Fubini's Theorem belong to $\bar{\cP}$.
	Since $\lambda$ is arbitrary and $G\cdot\hat{S}_T$ depends only on the variable $(\omega,\theta)$, we deduce that $(y,\bar{H})$ cannot be a superhedge.
	 \end{proof}
\begin{proof}[proof of Theorem \ref{thm:options}]
	From Lemma \ref{lem:options} and as in the proof of Theorem \ref{thm:main}, we have
	\[
	\sh_{\K,\Phi}(G)=\bar{\pi}_{\Phi}(G\cdot\hat{S}_T)=\bar{\pi}(G\cdot\hat{S}_T)=\bar{\pi}^r(G\cdot\hat{S}_T)=\sup_{\bar{\Q}\in\bar{\cQ}}\E_{\bar{\Q}}[G\cdot \hat{S}_T],
	\]
	where the set $\bar{\cQ}$ is the analogous, for $\bar{S}$, of the set $\hat{\cQ}$.
	To conclude note that, for any $\bar{\Q}\in\bar{\cQ}$, $\E_{\bar\Q}[\phi_k\cdot\hat{S}_T]< p_k$ for any $k=d+1,\ldots, m$.
	In particular, this implies $\E_{\bar\Q}[\varphi_j\cdot \hat{S}_T]\in(b_j,a_j)$ for any $j=1,\ldots e$ and, together with Proposition \ref{prop:eqDual}, the thesis follows.
	 \end{proof}

\appendix\normalsize
\section{Technical Results}
The following is a simple Lemma for convex sets in $\R^k$.
For $j=1,\ldots, k$, let $e_j$ be the $j$-th element of the canonical basis, $A_{j}:=\bigcup_{n\in\N} A+\frac{1}{n} e_j$ and $A_{-j}:=\bigcup_{n\in\N} A-\frac{1}{n} e_j$.
\begin{lemma}\label{lem:interConvex}
	Let $A$ be a convex set in $\R^k$ with $\inter(A)\neq\emptyset$.
	\[
	\bigcap_{j=1}^{k}\left(A_j\cap A_{-j}\right)=\inter(A).
	\]
\end{lemma}
\begin{proof}
	$(\supset)$.
	Let $x\in \inter{A}$. 
	For any $j=1,\ldots,k$, there exists $n_j$ such that $x+\frac{1}{n_j}e_j\in A$.
	Thus, $x\in A_{-j}$.
	In an analogous way we can show that $x\in A_j$ and the claim follows.
	
	$(\subset)$.
	Let $x\notin \inter{A}$.
	By a change of coordinate suppose $x=0$.
	Since $A$ is convex, by the Hyperplane Separation Theorem, there exists $H\in\R^k\setminus\{0\}$ such that $H\cdot y\ge 0$ for any $y\in A$.
	Let $j$ such that $H^j\neq 0$ and suppose that $H^j< 0$ (the case $H^j> 0$ follows analogously).
	From $H^j\cdot\frac{1}{n} e_j<0$ for any $n\in\N$, we have that $x=0\notin A_{-j}$.
	This concludes the proof.
	 \end{proof}
We provide here the proofs of Lemma \ref{lem:analytic_construction} and Proposition \ref{prop:nonempty} of Section \ref{sec:main}.
\begin{lemma}\label{lem:sum_graph}Let $\Omega$ be Polish and $\Phi,\Psi$ random sets in $\R^k$ with analytic graphs. 
	Then, $\Phi+\Psi$ has analytic graph. 
\end{lemma}
\begin{proof}
	Consider the function $f(\omega,\tilde{\omega},x,y):=(\omega,\tilde{\omega},x+y)$ for $(\tilde{\omega},\omega,x,y)\in\Omega\times\Omega\times \R^{k}\times \R^{k}$.
	$f(\gr(\Phi),\gr(\Psi))$ is the image of analytic sets through a Borel  function and thus, it is analytic.
	Moreover, the set $\{(\omega,\omega)\in\Omega\times\Omega\}$ is a Borel subset of $\Omega\times\Omega$.
	To conclude we observe that $\gr(\Phi+\Psi)$ is the projection on $\Omega\times \R^k$ of the analytic set
	\[
	f(\gr(K^*_t),\gr(C^*_t))\cap \big(\{(\omega,\omega)\in\Omega\times\Omega\}\times \R^k\big).
	\]
	 \end{proof}
\begin{proof}[proof of Lemma \ref{lem:analytic_construction}]
	By assumption $K_t$ is Borel measurable, thus, by \cite[Lemma 13]{BayraktarZhangMOR}, $\gr(K^*_t)$ is analytic.
	The result for $t=T$ follows. 
	We proceed by backward induction. 	
	Recall that $\gr(C^*_t)$ is analytic by assumption and, from the proof of \cite[Lemma 6]{BayraktarZhangMOR}, $\Gamma_t$ has analytic graph.
	From \cite[Lemma 12(b)]{BayraktarZhangMOR} the same is true for $\overline{\conv (\Gamma_t)}$.
	We conclude by using Lemma \ref{lem:sum_graph} and \cite[Lemma 12(d)]{BayraktarZhangMOR}.
	 \end{proof}
\begin{proof}[proof of Proposition \ref{prop:nonempty}]
	Define the market $\K^t:=\{K_0,\ldots, K_{t-1},\tilde{K}_t,\ldots,\tilde{K}_T\}$ for $t\in\cI$.
	We proceed by backward induction.
	For $t=T$, $\K^T=\K$ and the two properties follow by assumption.
	Suppose that the thesis is true for $u=t+1,\ldots T$, we show that it is true for $t$.
	For $\omega\in\Omega_t$, let 
	\[
	\Lambda_t(\omega):=\{x\in\R^d: x\in\tilde{K}_{t+1}(\omega;\cdot)\ \cP_t\text{-q.s.}\}.
	\]
	From \cite[Lemma 7]{BayraktarZhangMOR}, $\Lambda_t=\Gamma^*_t$, implying  $\Lambda^*_t\supset \tilde{K}^*_{t+1}$ $\cP$-q.s. 
	From the inductive hypothesis and Assumption \ref{ass:EF}, $\inter(\Lambda^*_t)\neq\emptyset$ and $\inter(K^*_t-C^*_t)\neq\emptyset$ outside a $\cP$-polar set $N$.
	Note now that for every $\omega\in N^C$ such that 
	\begin{equation}\label{eq:recursive_sep}
	\inter(\Lambda^*_t)(\omega)\cap\inter(K^*_t-C^*_t)(\omega)=\emptyset,
	\end{equation}
	we can find $x\in\R^d\setminus\{0\}$ such that $x\cdot y\le 0$ for any $y\in K^*_t-C^*_t$ and $x\cdot z\ge 0$ for any $z\in \Lambda^*_t$.
	From \cite[Lemma 16]{BayraktarZhangMOR} and since $K_t$ and $C_t$ are closed sets, we deduce that $x\in -K_t\cap C_t$ and $x\in \Lambda_t$.
	Let $\eta$ be a $\cF^u_t$-measurable selector of $\{ ( -K_t\cap C_t\cap\Lambda_t) \setminus\{0\}\}$ (for its existence, see \cite[Proposition 4]{BayraktarZhangMOR}).
	Note that $\eta\in A_{t+1}(\K^{t+1})$ and, by definition of $\Lambda_t$, $\eta\in \tilde{K}_{t+1}$ $\cP$-q.s.
	The strict no arbitrage condition implies that  $\eta=0$ $\cP$-q.s. and, as a consequence, the set of $\omega\in N^C$ such that \eqref{eq:recursive_sep} is satisfied is $\cP$-polar. 
	On the complementary set, we clearly have $\inter(\tilde{K}^*_t)\neq\emptyset$.
	We are only left to show that $\NAP$ holds for $\K^t$.
	Let $\eta:=(\eta_0,\ldots,\eta_T)$ and $r\ge t$ such that $\eta_r\in A_r(\K^t)\cap\cL^0(\cF^u_r;\tilde{K}_r)$ (the case $r\le t-1$ is trivial).
	By admissibility,
	\[\eta_r=-k_0-\ldots-k_{t-1}-\tilde{k}_t-\ldots -\tilde{k}_r,
	\]
	with $ k_s\in K_s$ for $s=0,\ldots t-1$ and $\tilde{k}_s\in \tilde{K}_s$ for $s=t,\ldots r$.
	Note that, for $t\le s\le T$, $\inter(\tilde{K}^*_s)\neq\emptyset$ $\cP$-q.s., therefore,
	\[ \tilde{K}_s=\big(K^*_s\cap (\overline{\conv(\Gamma_s)}+C^*_s)\big)^*=K_s+(\Lambda_s\cap C_s),
	\]
	where the first equality follows by \eqref{def:tildeK} and the second equality follows from \cite[Lemma 16]{BayraktarZhangMOR}.
	We start with the case $r\ge t+1$. From \cite[Lemma 8]{BayraktarZhangMOR}, we have $\tilde{k}_t=k_t+\lambda_{t}$ for some $k_t\in\cL^0(\cF^u_t;K_{t})$ and $\lambda_{t}\in\cL^0(\cF^u_{t};\tilde{K}_{t+1}\cap C_t)$.
	We can therefore rewrite 
	\[\eta_r=-k_0-\ldots-k_{t-1}-(k_t+\lambda_{t})-\tilde{k}_{t+1}-\ldots -\tilde{k}_r,
	\]
	Define the new strategy $\tilde{\eta}$ with $\tilde{\eta}_s=\eta_s$ for any $s\neq t$ and $\tilde{\eta}_t:=\eta_{t-1}-k_t=\eta_t+\lambda_{t}$.
	Since $\eta_t,\lambda_{t+1}\in\cL^0(\cF^u_{t};C_t)$ and $C_t$ is a convex cone, the sum takes also values in $C_t$, from which $\tilde{\eta}_t\in A_t(\K^{t+1})$.
	In particular, $\tilde{\eta}$ is admissible for the market $\K^{t+1}$ and satisfies $\tilde{\eta}_r\in A_r(\K^{t+1})\cap\cL^0(\cF^u_r;\tilde{K}_r)$ From the inductive hypothesis of strict no arbitrage, it follows that $\tilde{\eta}_r=\eta_r=0$ $\cP$-q.s.
	
	For the case $r=t$, by assumption $\eta_t\in\cL^0(\cF^u_r;\tilde{K}_t)$.
	Similarly as above, we can rewrite $\eta_t=\xi_t+\lambda_{t}$ for some $\xi_t\in\cL^0(\cF^u_t;K_{t})$ and $\lambda_{t}\in\cL^0(\cF^u_{t};\tilde{K}_{t+1}\cap C_t)$ and, hence, $\lambda_{t}=-k_0-\ldots-k_{t-1}-\xi_t$. This implies \[
	\lambda_{t}\in A_t(\K^{t+1})\cap \cL^0(\cF^u_t;\tilde{K}_{t+1})\subset A_{t+1}(\K^{t+1})\cap \cL^0(\cF^u_{t+1};\tilde{K}_{t+1}),
	\]
	where the inclusion follow from Assumption \ref{ass:EF}.
	The strict no arbitrage condition implies $\lambda_{t}=0$ $\cP$-q.s. and therefore $\eta_t=\xi_t$ $\cP$-q.s. Thus, $\eta_t\in \cA(\K^{t+1})\cap\cL^0(\cF^u_t;K_{t})$. Using again the strict no arbitrage condition, it follows $\eta_t=0$. 
	 \end{proof}

\bibliographystyle{abbrvnat}
\bibliography{Robust_bib}
\end{document}